\documentclass[journal, a4paper]{IEEEtran}

\usepackage{cite}

\ifCLASSINFOpdf

\else

\fi

\usepackage{graphicx}
\usepackage{manfnt}
\usepackage{times}
\usepackage{multicol}

\usepackage[cmex10]{amsmath}
\usepackage{amsmath,bm}
\usepackage{amssymb}
\DeclareMathOperator{\tr}{tr}
\DeclareMathOperator{\cov}{cov}
\DeclareMathOperator{\diag}{diag}
\DeclareMathOperator{\im}{Im}
\DeclareMathOperator{\Ker}{Ker}
\DeclareMathOperator{\ri}{ri}

\allowdisplaybreaks

\begin{document}

\sloppy

\title{Rate Region of the Vector Gaussian CEO Problem with the Trace Distortion Constraint}

\author{Yinfei~XU,
  ~\IEEEmembership{Student Member,~IEEE,}
  ~and~\ Qiao~WANG,
  ~\IEEEmembership{Member,~IEEE}

\thanks{The authors are with School of Information Science and Engineering, Southeast University, Nanjing 210096, China. Email: \{yinfeixu, qiaowang\}@seu.edu.cn}
}
\date{today}

\maketitle

\begin{abstract}
We establish a new extremal inequality, which is further leveraged to give a complete characterization of the rate region of the vector Gaussian CEO problem with the trace distortion constraint. The proof of this extremal inequality hinges on a careful analysis of the Karush-Kuhn-Tucker necessary conditions for the non-convex optimization problem associated with the Berger-Tung scheme, which enables us to integrate the perturbation argument by Wang and Chen with the distortion projection method by Rahman and Wagner.
%This paper characterizes the rate region of the vector Gaussian CEO problem with total average quadratic distortion. We develop a new analysis technique based on spectral decomposition of mean square error in Berger-Tung scheme in order to prove the converse part of the rate distortion region, in which the perturbation method is utilized through combined with detailed analysis of Karush-Kuhn-Tucker necessary conditions in non-convex optimization problem. Finally, we show that Berger-Tung inner bound can achieve the entire rate region of the vector Gaussian CEO problem, via deriving a novel extremal inequality.
\end{abstract}


\IEEEpeerreviewmaketitle
\section{Introduction}
\IEEEPARstart{T}HE CEO problem, which is a special case of multi-terminal source coding, was first investigated by Berger, Zhang and
Viswanathan \cite{CEO}. Oohama \cite{Oohama} determined the asymptotic sum-rate-distortion function of the scalar Gaussian CEO problem via an ingenious application of the entropy power inequality. A complete characterization of the rate region of the scalar Gaussian CEO was obtained in \cite{regionCEO2} and  \cite{regionCEO}. However, extending this result to the vector case is not straightforward due to the fact that the entropy power inequality is not necessarily tight in this setting. Tavilder and  Viswanath \cite{tav} derived a lower bound on the sum rate of the vector Gaussian CEO problem by partially replacing the entropy power inequality with the worst additive noise lemma.  An explicit lower bound on the weighted sum rate of the two-terminal vector Gaussian CEO problem can be found in \cite{v1}. Of particular relevance here is the work by Wang and Chen \cite{v2,v3}, where they derived an outer bound on the rate region of the vector Gaussian CEO problem by establishing a certain extremal inequality; essentially the same result was obtained independently by Ekrem and Ulukus via exploiting the
 relation between Fisher information matrix and MMSE (minimum mean square error) \cite{v4}. The extremal inequality in \cite{v2,v3} is a variant of the Liu-Viswanath inequality \cite{exinq}, which is in turn inspired by the seminal work of Weingarten, Steinberg  and Shamai \cite{MIMO} on the characterization of the capacity region of the MIMO
Gaussian broadcast channel.

However, the outer bound induced by the Wang-Chen extremal inequality is in general not tight. Our main result is a strengthened extremal inequality for the special case where the covariance distortion constraint is replaced with the trace distortion constraint. It turns out that this new extremal inequality yields a complete characterization of the rate region of the vector Gaussian CEO problem for this special case. The perturbation argument, which is widely used for establishing extremal inequalities, appears to be insufficient for our purpose. For this reason, we develop a spectral decomposition method, which can be effectively incorporated into the perturbation argument to obtain the desired inequality. It is worth mentioning that our spectral decomposition method is partly motivated by the \emph{distortion projection} technique developed by Rahman and Wagner \cite{scalar,vector}  for  the vector Gaussian one-help-one problem (see also \cite{my} for a direct proof based on the perturbation method). %In this sense, our main conceptual contribution is the integration of the perturbation method with the projection method.

% However, these works yield only the outer bound, which is not necessarily tight, and the entire rate distortion region of the vector Gaussian CEO problem is still unknown.\par

%AFTER the seminal work of Weingarten, Steinberg  and Shamai \cite{MIMO} which characterizes the capacity region of MIMO
%Gaussian broadcast channel, several current works turn attention to characterize the rate distortion region from scalar Gaussian source to vector Gaussian source in multi-terminal source coding problem.\par
\newcounter{MYtempeqncnt}
%\cite{exinq, scalar, vector,my,v1, v2, v3, v4}.
\par
The rest of this paper is organized as follows. In Section \uppercase\expandafter{\romannumeral 2},
we present the formulation of the vector Gaussian CEO problem under the trace distortion constraint and the corresponding Berger-Tung upper bound on the weighted sum rate.
 In Section \uppercase\expandafter{\romannumeral 3}, we revisit some mathematical preliminaries which will be
used frequently in our proof. In Section \uppercase\expandafter{\romannumeral 4}, we
   prove certain properties of the spectral decomposition of the mean squared error matrix of the Berger-Tung scheme based on a carefully analysis of the KKT conditions of an associated non-convex optimization problem. In Section \uppercase\expandafter{\romannumeral 5},  we establish a new extremal inequality by considering projections into subspaces specified by the spectral decomposition result in the previous section, which is further leveraged to characterize the rate region of the vector Gaussian CEO problem with the trace distortion constraint.
 Finally,
  we conclude this paper in Section\uppercase\expandafter{\romannumeral 6}.\par

\section{Problem Statement and the Main Result}

\newcounter{tmp}

The system model of the vector Gaussian CEO problem is depicted in Figure \ref{fig}. Let $\{ \mathbf{X}(t)\}_{t=1}^{\infty}$ be an $m \times 1$-dimensional i.i.d. vector-valued sequence, where each $ \mathbf{X}(t), t=1, 2, \ldots$ is
a Gaussian random vector with mean zero and covariance $\mathbf{K} \succ 0$. For $i=1,2,\ldots, L$, let
\begin{equation}
\mathbf{Y}_{i}(t)=\mathbf{X}(t)+\mathbf{N}_{i}(t), \qquad i=1,2,\ldots,L \nonumber
\end{equation}
where $\mathbf{N}_{i}(t), t=1,2, \ldots$ are i.i.d. Gaussian random $m \times 1$-dimensional vectors independent of $\{\mathbf{X}(t)\}_{t=1}^{\infty}$
with mean zero and covariance $\mathbf{\Sigma}_{i} \succ 0 $. The noise processes $\{\mathbf{N}_{i}(t)\}_{t=1}^{\infty}$, $i=1,2,\ldots,,L$, are mutually independent. %For simplicity, we  denote $n$-length sequence by a superscript $n$ for simplicity,
% e.g., $\mathbf{X}^{n}=\{\mathbf{X}(1), \cdots, \mathbf{X}(n)\}$ and  $\mathbf{Y}_{i}^{n}= \{\mathbf{Y}_{i}(1), \cdots, \mathbf{Y}_{i}(n) \}$, for $i=1,2,\cdots, L$.
For $i=1,2,\cdots, L$, encoder $i$ computes $C_{i} = \phi_{i}^{n}(\mathbf{Y}_{i}^{n})$ based on its noisy observation $\mathbf{Y}_{i}^{n}= \{\mathbf{Y}_{i}(1), \cdots, \mathbf{Y}_{i}(n) \}$ using encoding function
 \begin{equation}
\phi_{i}^{n}: \mathcal{R}^{m\times n} \mapsto \mathcal{M}_{i}^{n} = \{1, \cdots,  2^{n R_{i}} \} \nonumber
\end{equation}
and sends $C_i$ to the decoder. Upon receiving $C_1,C_2,\ldots,C_L$, the decoder computes  $\hat{\mathbf{X}}^{n} =\{\hat{\mathbf{X}}(1), \cdots, \hat{\mathbf{X}}(n)\}= \varphi^{n}(C_{1}, \cdots, C_{L})$, which is an estimate of the remove source $\mathbf{X}^{n}=\{\mathbf{X}(1), \cdots, \mathbf{X}(n)\}$, using decoding function
\begin{equation}
\varphi^{n}: \mathcal{M}_{1}^{n} \times \ldots \times \mathcal{M}_{L}^{n} \mapsto \mathcal{R}^{m\times n}. \nonumber
\end{equation}

%bserves sequence $\mathbf{Y}_{i}^{n}\triangleq \{\mathbf{Y}_{i}(1), \cdots, \mathbf{Y}_{i}(n) \}$ and separately uses itself encoding function:
%\begin{equation}
%\phi_{i}^{n}: \mathcal{R}^{m\times n} \mapsto \mathcal{M}_{i}^{n} = \{1, \cdots,  2^{n R_{i}} \} \nonumber
%\end{equation}
%i.e., $C_{i} = \phi_{i}^{n}(\mathbf{Y}_{i}^{n})$, and sends message $C_{i}$ to the decoder. The decoder receives $L$ messages and estimates
% $\mathbf{X}^{n}\triangleq\{\mathbf{X}(1), \cdots, \mathbf{X}(n)\}$ using itself decoding function:
%\begin{equation}
%\varphi^{n}: \mathcal{M}_{1}^{n} \times \ldots \times \mathcal{M}_{L}^{n} \mapsto \mathcal{R}^{m\times n} \nonumber
%\end{equation}
%i.e. $\hat{\mathbf{X}}^{n} = \varphi^{n}(C_{1}, \cdots, C_{L})$.\par

 Throughout the paper, we adopt the trace distortion constraint. Specifically, a rate tuple $(R_{1}, \ldots, R_{L}, d)$ is said
 to be achievable subject to the trace distortion constraint $d$ if there exist encoding functions $\phi_{1}^{n}, \ldots, \phi_{L}^{n}$ and decoding function $\varphi^{n}$ such that

\begin{figure}
\centering
\includegraphics[width=0.5\textwidth]{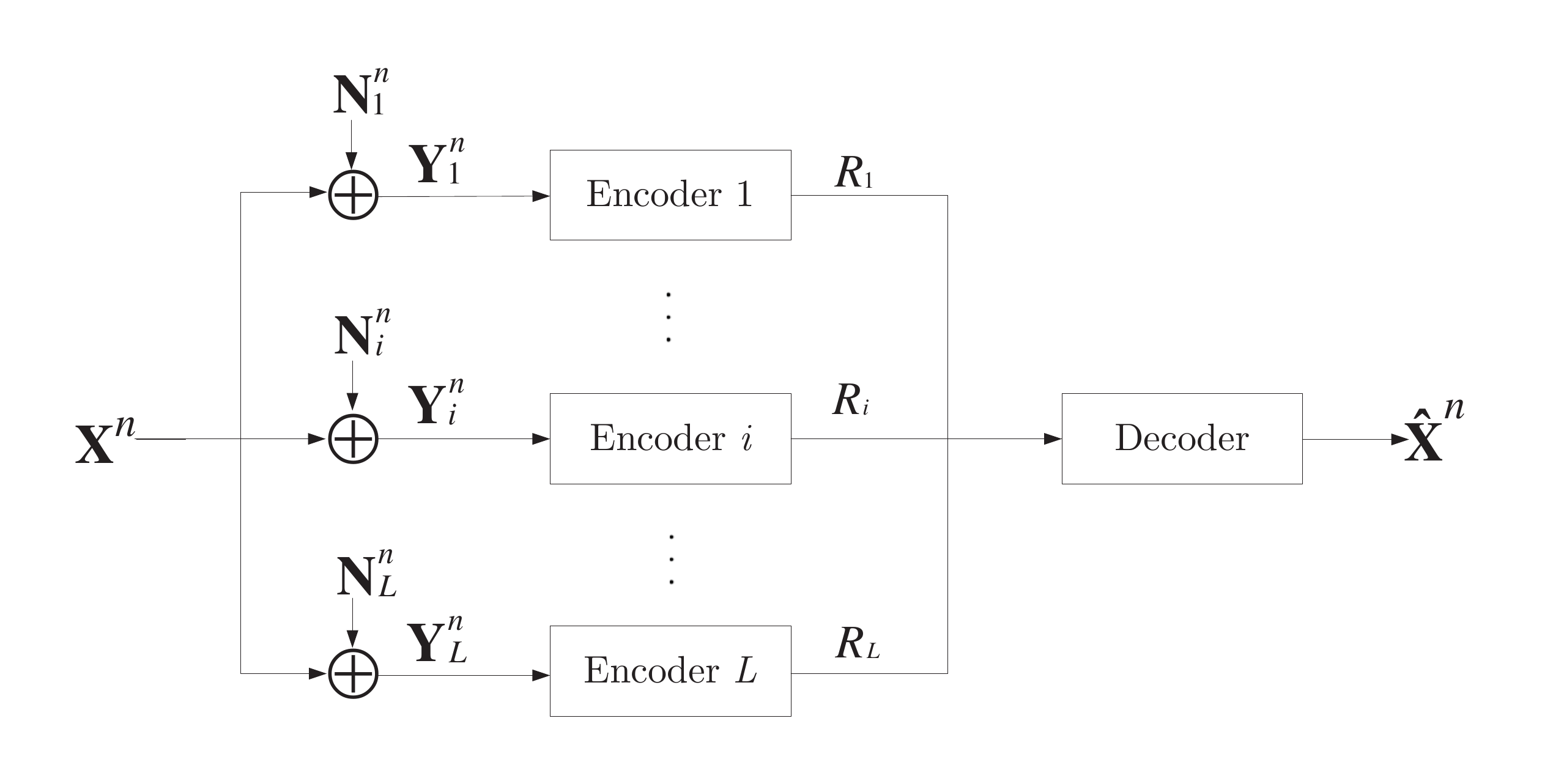}
\caption{The vector Gaussian CEO problem with trace constraint $\tr\{\cov(\hat{\mathbf{X}}-\mathbf{X})\} \leq d$.}
\label{fig}
\centering
\end{figure}

$$
\frac{1}{n} \sum_{t=1}^{n}  \mathbb{E} \left[ \tr(\cov(\mathbf{X}(t)-\hat{\mathbf{X}}(t)))\right] \leq d,
$$
 where $\mathbf{x}_{j}(t)$ and $\hat{\mathbf{x}}_{j}(t)$ represents the $j$-\emph{th} component of random vectors
 $\mathbf{X}(t) \triangleq \left( \mathbf{x}_{1}(t), \ldots, \mathbf{x}_{m}(t)\right)^{T}$ and
  $\hat{\mathbf{X}}(t) \triangleq \left( \hat{\mathbf{x}}_{1}(t), \ldots, \hat{\mathbf{x}}_{m}(t)\right)^{T}$ respectively. The rate region
   $\mathcal{R}(d)$ is the closure of all achievable rate tuples $(R_{1},\cdots, R_{L})$ subject to the trace distortion constraint $d$. \par
Since the rate region is convex, it can be characterized  by its supporting hyper-planes. As a consequence, it suffices to solve the following optimization problem
$$
R(d) = \inf_{(R_{1}, \ldots, R_{L)} \in \mathcal{R}(d)} \sum_{i=1}^{L} \mu_{i} R_{i}
$$
for $\mu_{i} \geq 0, i=1, \ldots, L$; moreover, there is no loss of generality in assuming $\mu_{1} \geq \cdots \geq \mu_{L} \geq 0$.  Note that if $\mu_L=0$, then one can reduce the $L$-terminal problem to the $(L-1)$-terminal problem by providing $\mathbf{Y}_{L}^{n}$ directly to the decoder and the first $L-1$ encoders. For this reason, we shall focus on the case  $\mu_{1} \geq \cdots \geq \mu_{L} > 0$ in the rest of this paper.

It is clear that $R(d) = \infty $ when $d \leq  \tr \{ (\mathbf{K}^{-1}+ \sum_{i=1}^{L} \mathbf{\Sigma}_{i})^{-1} \}$,  and  $R(d)=0$ when $d \geq \tr \{ \mathbf{K}\}$. Henceforth the only case
\begin{equation}
\tr \{ (\mathbf{K}^{-1}+ \sum_{i=1}^{L} \mathbf{\Sigma}_{i})^{-1} \} < d <  \tr \{ \mathbf{K}\}
\end{equation}
needs to be considered.

By evaluating the standard
Berger-Tung scheme, one can readily show that
\begin{align*}
R(d)\leq R^{BT}(d),
\end{align*}
where
\begin{align}
R^{BT}(d)= \min_{(\mathbf{B}_{1}, \cdots, \mathbf{B}_{L})} \sum_{i=1}^{L-1} \frac{\mu_{i}-\mu_{i+1}}{2} \log \frac{|\mathbf{K}^{-1}+\sum_{j=1}^{L}\mathbf{B}_{j}|}{|\mathbf{K}^{-1}+\sum_{j=i+1}^{L} \mathbf{B}_{j}|} \nonumber \\
+ \sum_{i=1}^{L} \frac{\mu_{i}}{2} \log{\frac{ | \mathbf{\Sigma}_{i}^{-1}| }{ |\mathbf{\Sigma}_{i}^{-1}-\mathbf{B}_{i}| }} + \frac{\mu_{L}}{2} \log \frac{|\mathbf{K}^{-1}+\sum_{j=1}^{L}\mathbf{B}_{j}|}{|\mathbf{K}^{-1}|}.\label{eq:BT}
\end{align}
The minimization in (\ref{eq:BT}) is over $(\mathbf{B}_{1}, \ldots, \mathbf{B}_{L})$ subject to constraints
\begin{align}
&\tr\left \{\left (\mathbf{K}^{-1}+ \sum_{i=1}^{L} \mathbf{B}_{i}\right )^{-1} \right \}\leq d,  \label{eq:traceconstraint}\\
&\mathbf{\Sigma}_{i}^{-1} \succeq \mathbf{B}_{i} \succeq 0, \qquad i=1, \ldots, L. \nonumber
\end{align}
%\newtheorem{remark}{Remark}
%\begin{remark}
%This inner bound is equivalent to the one in\cite{v3,v4}, expect for the trace constraint of the total
% average distortion $d$. Alternatively, the authors of in\cite{v3,v4} considered the covariance constraint of positive definite matrix distortion $\mathbf{D}$.
%\end{remark}
%\smallskip

The main result of this paper is the following theorem.

\newtheorem{theorem}{Theorem}
\begin{theorem}\label{thm1}
For any $\mu_{1} \geq \cdots \geq \mu_{L}> 0$ and $d\in(\tr \{ (\mathbf{K}^{-1}+ \sum_{i=1}^{L} \mathbf{\Sigma}_{i})^{-1} \} , \tr ( \mathbf{K}))$,
$$R(d) = R^{BT}(d).$$
\end{theorem}
\smallskip

The rest of this paper is devoted to the proof of the converse part of the theorem, i.e.,
$$
R(d) \geq R^{BT}(d).
$$

\section {Mathematical Preliminaries}
We first review some basic properties of conditional Fisher Information Matrix and MSE (mean square error).
\newtheorem{definition}{Definition}
\begin{definition}
Let $(\mathbf{X}, U)$ be a pair of jointly distributed random vectors with differentiable conditional probability density function $f(\mathbf{x}|u)$.
The vector-valued score function is defined as
$$\nabla \log f(\mathbf{x}| u)  = \left[\frac{\partial \log f(\mathbf{x}| u)}{\partial x_{1}}, \cdots, \frac{\partial \log f(\mathbf{x}| u)}{\partial x_{m}} \right]^{T}.$$
The conditional Fisher Information of $\mathbf{X}$ respect to $U$ is given by
$$
J(\mathbf{X}| U) = \mathbb{E}\left[  \left(\nabla \log f(\mathbf{x}| u) \right) \cdot \left(\nabla \log f(\mathbf{x} | u) \right)^{T}    \right ].
$$
\end{definition}
\newtheorem{lemma}{Lemma}

\smallskip
\begin{lemma} [Cram\'{e}r--Rao Lower Bound] \label{lea1}
Let $(\mathbf{X}, U)$ be a pair of jointly distributed random vectors. Assuming that the conditional covariance matrix $ \cov (\mathbf{X}| U) \succ \mathbf{0}$, then
\begin{equation}
J(\mathbf{X}| U)^{-1} \preceq \cov (\mathbf{X}| U). \label{lemma1}
\end{equation}
\end{lemma}
One can refer to the proof in \cite[Appendix \uppercase\expandafter{\romannumeral 2}]{exinq}.
\smallskip

\begin{lemma}[Complementary Identity] \label{lea2}
Let $(\mathbf{X}, \mathbf{N},U)$ be a tuple of jointly distributed random vectors. If $\mathbf{N}$ follows a Gaussian distribution
 $\mathcal{N}(\mathbf{0}, \mathbf{\Sigma})$, and it is independent with $(\mathbf{X},U)$, then
\begin{equation}
 J(\mathbf{X} + \mathbf{N} | U) + \mathbf{\Sigma}^{-1} \cov(\mathbf{X} | \mathbf{X} + \mathbf{N}, U) \mathbf{\Sigma}^{-1} = \mathbf{\Sigma}^{-1} \label{lemma2}
 \end{equation}
\end{lemma}
The proof of this complementary identity can be found in \cite[Corollary 1]{Palomar}.
\smallskip

\begin{lemma} [de Bruijn's Identity]\label{lea3}
Let $(\mathbf{X}, U)$ be a pair of jointly distributed random vectors, and $\mathbf{N} \thicksim \mathcal{N}(\mathbf{0}, \mathbf{\Sigma})$ be a
 standard Gaussian random vector, which is independent of  $(\mathbf{X}, U)$, then
 \begin{equation}
 \frac{d}{d \gamma }h(\mathbf{X}+ \sqrt{\gamma} \mathbf{N} | U) = \frac{1}{2} \tr \left\{ J (\mathbf{X}+\sqrt{\gamma}\mathbf{N} | U) \mathbf{\Sigma}\right\}. \label{lemma3}
 \end{equation}
\end{lemma}
This lemma is the conditional version of \cite[Theorem 14]{Dembo}. %The de Bruijn's identity recovers the links between entropy and Fisher information matrix.\par

Replacing the variable $\gamma$ by ${1}/{\gamma}$ in de Bruijn's identity and using the complementary identity in lemma \ref{lea2},
one can obtain the following result via simple algebraic manipulations.
\smallskip
\newtheorem{corollary}{Corollary}
\begin{corollary} \label{cor1}
 \begin{equation}
 \frac{d}{d \gamma }h(\sqrt{\gamma}\mathbf{X}+ \mathbf{N} | U) = \frac{1}{2} \tr \left\{\mathbf{\Sigma}^{-1}\cov (\mathbf{X}|\sqrt{\gamma} \mathbf{X}+\mathbf{N}) \right\}. \label{corollary1}
 \end{equation}
\end{corollary}
\smallskip
\begin{lemma} [Data Processing Inequality] \label{lea6}
Let $(\mathbf{X}, U, V)$ be a tuple of jointly distributed random vectors, and $U, V, \mathbf{X}$ form a Markov chain. \emph{i.e.} $U \rightarrow V \rightarrow \mathbf{X}$, then
\begin{equation}
J(\mathbf{X} | U) \preceq J(\mathbf{X} | V).
\end{equation}
\end{lemma}
The proof follows easily by the chain rule of Fisher information matrix \cite[Lemma 1]{lemma2}.%, which is analogous to
%another form of data processing inequality \cite[Lemma 3]{lemma2}.
\smallskip

\begin{lemma} [Fisher Information Inequality]\label{lea4}
Let $(\mathbf{X}, \mathbf{Y},U)$ be a tuple of jointly distributed random vectors. Assume that $\mathbf{X}$ and $\mathbf{Y}$ be
conditionally independent given $U$, then for any $\gamma \in (0,1)$,
\begin{equation}
J(\sqrt{1-\gamma}\mathbf{X}+ \sqrt{\gamma}\mathbf{Y} | U) \preceq (1-\gamma) J(\mathbf{X} | U) + \gamma J(\mathbf{Y} | U).\label{lemma4}
\end{equation}
\end{lemma}
This is an equivalent form of matrix Fisher information inequality. One can refer to \cite[Proposition 3]{EPI} for a detailed discussion.
\smallskip

\begin{lemma} \label{lea5}
Let $(\mathbf{X}, U)$ be a pair of jointly distributed random vectors, and $\mathbf{N} \thicksim \mathcal{N}(\mathbf{0}, \mathbf{\Sigma})$ be a
 standard Gaussian random vector, which is independent of  $(\mathbf{X}, U)$, then for any $\gamma \in (0,1)$, we have
\begin{equation}
\cov (\mathbf{X} | \mathbf{X} + \mathbf{N}, U) \preceq \gamma^{2} \cov (\mathbf{X}| U) +  (1-\gamma)^{2}\mathbf{\Sigma}. \label{lemma5}
\end{equation}
\end{lemma}
The proof is left in Appendix \ref{OP}.

\section {Properties of $R^{BT}(d)$}
In this section, we study the KKT (Karush--Kuhn--Tucker) conditions for the optimization problem $R^{BT}(d)$ and establish some basic properties of the subspaces induced by the eigen-decomposition of the MSE (mean square error) matrix. These properties play a key role in the proof of the converse theorem for the vector Gaussian CEO problem under the trace distortion constraint. \par

\subsection{KKT Conditions}
It is easy to observe that the objective function
of the optimization problem $R^{BT}(d)$ goes to infinity as $|\mathbf{\Sigma}_{i}^{-1}-\mathbf{B}_{i}| \rightarrow 0$ for any $i=1,\ldots, L$. Hence the constraints $\mathbf{B}_{i} \preceq \mathbf{\Sigma}_{i}^{-1}$, $i=1, \ldots, L$
are not active. \par
The Lagrangian of the optimization problem $R^{BT}(d)$ is given by
\begin{align}
\frac{\mu_{1}}{2} \log {\left |\mathbf{K}^{-1}+\sum_{j=1}^{L}\mathbf{B}_{j}\right |} - \sum_{i=1}^{L-1} \frac{\mu_{i}-\mu_{i+1}}{2}\log\left |\mathbf{K}^{-1} + \sum_{j=i+1}^{L} \mathbf{B}_{j}\right | \nonumber \\
- \sum_{i=1}^{L} \frac{\mu_{i}}{2} \log |\mathbf{\Sigma}_{i}^{-1} - \mathbf{B}_{i}| + \sum_{i=1}^{L} \frac{\mu_{i}}{2} \log |\mathbf {\Sigma}_{i}^{-1}| - \frac{\mu_{L}}{2} \log |\mathbf{K}^{-1}| \nonumber \\
-  \sum_{i=1}^{L}\tr \left({\mathbf{B}_{i} \mathbf{\Psi}_{i}}\right) + \lambda \left(\tr \left((\mathbf{K}^{-1} + \sum_{j=1}^{L} \mathbf{B}_{j})^{-1} \right) -d\right), \nonumber
\end{align}
where matrices $\mathbf{\Psi}_{i}, i=1, \ldots, L$ and scalar $\lambda$ are Lagrange multipliers.  Let $\mathbf{B}_{1}^{*}, \ldots, \mathbf{B}_{L}^{*}$ be the optimal solution of $R^{BT}(d)$. Define
\begin{equation}\label{eqn:KB-operator}
\mathbf{C}_i=\left( \mathbf{K}^{-1} + \sum_{j=i}^{L} \mathbf{B}_{j}^{*}\right)^{-1},\ \ i=1,2,\cdots,L.
\end{equation}
The KKT conditions for the optimization problem $R^{BT}(d)$ are given by
\begin{equation} \label{KKT1}
\frac{\mu_{1}}{2} \mathbf{C}_1 + \frac{\mu_{1}}{2} \left( \mathbf{\Sigma}_{1}^{-1} -  \mathbf{B}_{1}^{*} \right)^{-1} - \mathbf{\Psi}_{1} - \lambda \mathbf{C}_1^2 =0;
\end{equation}
\begin{align} \label{KKT2}
\frac{\mu_{1}}{2} \mathbf{C}_1  + \frac{\mu_{k}}{2}  \left( \mathbf{\Sigma}_{k}^{-1} -  \mathbf{B}_{k}^{*} \right)^{-1}\nonumber \\
 - \sum_{i=1}^{k-1} \frac{\mu_{i}-\mu_{i+1}}{2} \mathbf{C}_{i+1} - \mathbf{\Psi}_{k}
 - \lambda \mathbf{C}_1^2 =0; \nonumber \\
 k=2, \ldots, L;
\end{align}
\begin{align}
\mathbf{B}_{k}^{*} \mathbf{\Psi}_{k} =0, \quad k=1, \ldots, L;  \label{KKT3}\\
\lambda \left(\tr \left(\mathbf{C}_1 \right) -d\right)=0; \label{KKT4}\\
 \quad \mathbf{\Psi}_{k} \succeq 0, \quad k=1,\ldots, L; \quad \lambda \geq 0. \label{KKT5}
\end{align}
\par
Notice that the optimization problem $R^{BT}(d)$ is not convex; therefore, the constraint qualifications need to be examined in order to show the existence of  Lagrange multipliers $\mathbf{\Psi}_{i}, i=1, \ldots, L$ and $\lambda$
satisfying  the KKT conditions. These technical details are relegated to Appendix \ref{KKT}. Here we just point out the following implication of the KKT conditions.
\begin{corollary} \label{lamd}
For  $d\in(\tr \{ (\mathbf{K}^{-1}+ \sum_{i=1}^{L} \mathbf{\Sigma}_{i})^{-1} \} , \tr ( \mathbf{K}))$, we have
\begin{equation}
\tr \left(\mathbf{C}_1 \right) = d. \label{lambda}
\end{equation}
\end{corollary}
\begin{proof}
According to the complementary slackness condition \eqref{KKT4},
for the purpose of proving \eqref{lambda}, it suffices to show $\lambda \neq 0$. If $\lambda = 0$, then it follows by \eqref{KKT1} that $\mathbf{\Psi}_{1} \succ {0}$, which, together with  the complementary slackness condition $\mathbf{B}_{1}^{*} \mathbf{\Psi}_{1}=0$ in \eqref{KKT3}, implies $\mathbf{B}_{1}^{*} =0$. Substituting $\mathbf{B}_{1}^{*} =0$ into the first equation in \eqref{KKT2} gives $\mathbf{\Psi}_{2} \succ {0}$. Along this way, we may inductively obtain
$\mathbf{B}_{1}^{*} = \mathbf{B}_{2}^{*} = \ldots =\mathbf{B}_{L}^{*} = 0 $, which, in view of (\ref{eq:traceconstraint}), implies $\tr(\mathbf{K}) \leq d$. This leads to a contradiction with the assumption $d<\tr\{\mathbf{K}\}$. Thus \eqref{lambda} is proved.
\end{proof}
%%Because the procedure is very technique and similar to the proof in\cite[Appendix \uppercase\expandafter{\romannumeral 4}]{MIMO},\cite[Appendix B]{vector}. We omit it in this paper.

\subsection{Spectral-decomposition of MSE}
Since the mean square error matrix $\mathbf{C}_1=( \mathbf{K}^{-1}+ \sum_{j=1}^{L}\mathbf{B}_{j}^{*})^{-1}$ of the Berger-Tung scheme is positive definite,
we can write its spectral representation as below:
\begin{equation}\label{eqn:SD-2} \mathbf{C}_1= \sum_{n=1}^{m} d_{n} \bm{e}_{n}\bm{ e} _{n}^{T},\end{equation}
 where the positive real
 numbers $d_{n}, n=1,\cdots,m $ stand for the eigenvalues, and $ \bm{e}_{1}, \bm{e}_{2}, \ldots, \bm{e}_{m} \in \mathbb{R}^{m}$ are the corresponding normalized eigenvectors which form an orthogonal basis.\par
 It follows readily  from \eqref{eqn:SD-2}  that
 \begin{equation}\label{eqn:SD-3}
\mathbf{C}_1^{2}= \sum_{n=1}^{m} d_{n}^{2} \bm{e}_{n}\bm{ e} _{n}^{T}.\end{equation}
In what follows,  we denote
\begin{equation}\label{eqn:SD-1}\mathbf{\Delta}_{i} \triangleq  \frac{\mu_{i}}{2}( \mathbf{\Sigma}_{i}^{-1} -  \mathbf{B}_{i}^{*})^{-1}-\mathbf{\Psi}_{i}.\end{equation}
By the matrix identity in KKT conditions \eqref{KKT1}, we see that
 $$\mathbf{\Delta}_{1} =\lambda  \mathbf{C}_1^2-  \frac{\mu_{1}}{2}  \mathbf{C}_1,$$
 Substituting \eqref{eqn:SD-2} and \eqref{eqn:SD-3} into the above equation leads to the following spectral representation of $\mathbf{\Delta}_{1}$:
\begin{equation}\label{eqn:add-1}\mathbf{\Delta}_{1} = \sum_{n=1}^{m} \left (\lambda  d_{n}^{2} - \mu_{1}\frac{d_{n}}{2} \right ) \bm{e}_{n}\bm{ e} _{n}^{T}.\end{equation}
Now we divide the vector space $\mathbb{R}^{m}$ into two orthogonal subspaces according to the sign of the eigenvalues $\lambda  d_{n}^{2} - \mu_{1}d_{n}/2, n=1,2,\ldots,m $. We may define $m \times n_{1}$
matrix $\mathbf{U}_{1}  \triangleq  \left( \bm{e}_{1}, \bm{e}_{2}, \ldots, \bm{e}_{n_{1}}\right)$  in which the eigenvectors $\bm{e}_{n}, n=1, 2,\ldots, n_{1}$,
correspond to the positive eigenvalues. Similarly we may  define $m \times (m-n_{1})$ matrix $\mathbf{V}_{1}  \triangleq  \left( \bm{e}_{n_{1}+1}, \bm{e}_{2}, \ldots, \bm{e}_{m}\right)$,
 in which the eigenvectors $\bm{e}_{n}, n=n_{1}+1, n_{1}+2,\ldots, m $, correspond to  non-positive eigenvalues. It can be verified that
\begin{align}
&\mathbf{U}_{1}^{T} \mathbf{\Delta}_{1} \mathbf {U}_{1} \succ 0, \quad \mathbf{V}_{1}^{T} \mathbf{\Delta}_{1} \mathbf {V}_{1}\preceq 0, \quad  \mathbf{U}_{1}^{T} \mathbf{\Delta}_{1} \mathbf {V}_{1}=0; \\
&\mathbf{U}_{1}^{T} \mathbf{C}_{1}\mathbf {U}_{1} \succeq 0, \quad
 \mathbf{V}_{1}^{T}\mathbf{C}_{1} \mathbf {V}_{1} \succeq 0,  \quad
 \mathbf{U}_{1}^{T}\mathbf{C}_{1}\mathbf {V}_{1} = 0.
\end{align}
At this stage we may rewrite the spectral decomposition of $\mathbf{\Delta}_{1}$ and $ \mathbf{C}_1=( \mathbf{K}^{-1}+ \sum_{j=1}^{L}\mathbf{B}_{j}^{*})^{-1}$ according to the positivity/non-positivity structure of eigenspaces as below:
\begin{align}
 &\mathbf{\Delta}_{1} = \mathbf{U}_{1} \mathbf{U}_{1}^{T} \mathbf{\Delta}_{1} \mathbf {U}_{1} \mathbf{U}_{1}^{T} + \mathbf{V}_{1}\mathbf{V}_{1}^{T} \mathbf{\Delta}_{1} \mathbf {V}_{1}\mathbf{V}_{1}^{T}, \\
 &\mathbf{C}_{1} =\displaystyle \mathbf{U}_{1} \mathbf{U}_{1}^{T}\mathbf{C}_{1}\mathbf {U}_{1} \mathbf{U}_{1}^{T}  + \mathbf{V}_{1}\mathbf{V}_{1}^{T}\mathbf{C}_{1} \mathbf {V}_{1}\mathbf{V}_{1}^{T}.
\end{align}
\par
\smallskip
Since $ \mathbf{V}_{1}^{T} \mathbf{\Delta}_{1} \mathbf {V}_{1}\preceq 0 $, we have
$$ \mathbf{V}_{1}^{T} \mathbf{\Psi}_{1}\mathbf {V}_{1} \succeq  \frac{\mu_{1}}{2} \mathbf{V}_{1}^{T}( \mathbf{\Sigma}_{1}^{-1} -  \mathbf{B}_{1}^{*})^{-1} \mathbf {V}_{1} \succ 0,$$
 which means that the subspace spanned by the column vectors of $\mathbf{V}_{1}$ belongs to the image space of $\mathbf{\Psi}_{1}$, \emph{i.e.}, $\mathbf{V}_{1} \subseteq \im (\mathbf{\Psi_{1}})$.
 Thus by the complementary slackness conditions \eqref{KKT3} in KKT conditions, we have $\mathbf{B}_{1}^{*} \mathbf{\Psi}_{1} = 0$; as a consequence, the kernel space of $\mathbf{B}_{1}^{*}$  contains the image  space of $\mathbf{\Psi}_{1}$,  \emph{i.e.}, $ \Ker (\mathbf{B}_{1}^{*}) \supseteq \im (\mathbf{\Psi_{1}})$, which implies
\begin{equation}\label{eqn:b1}
\mathbf{B}_{1}^{*} \mathbf{V}_{1} =0.
\end{equation}
Henceforth, according to the definition of $\mathbf{V}_{1}$, we have
\begin{align} \label{orth}
0&= \mathbf{B}_{1}^{*}\mathbf{V}_{1}  \diag (d_{n_{1}+1}, d_{n_{1}+2}, \ldots, d_{m}) \nonumber \\
 &=  \mathbf{B}_{1}^{*} (\bm{e}_{n_{1}+1},\bm{e}_{n_{1}+2}, \ldots, \bm{e}_{m} ) \diag (d_{n_{1}+1}, d_{n_{1}+2}, \ldots, d_{m}) \nonumber \\
 &= \mathbf{B}_{1}^{*} (d_{n_{1}+1}\bm{e}_{n_{1}+1},d_{n_{1}+2}\bm{e}_{n_{1}+2}, \ldots, d_{m}\bm{e}_{m} ) \nonumber \\
 &= \mathbf{B}_{1}^{*} \mathbf{C}_{1} (\bm{e}_{n_{1}+1},\bm{e}_{n_{1}+2}, \ldots, \bm{e}_{m} ) \nonumber\\
 &= \mathbf{B}_{1}^{*} \mathbf{C}_{1} \mathbf{V}_{1}.
\end{align}
Left-multiplying with $\mathbf{C}_2=( \mathbf{K}^{-1}+ \sum_{j=2}^{L}\mathbf{B}_{j}^{*})^{-1} $ at both sides of \eqref{orth} yields
\begin{align}
0 & = \mathbf{C}_2  \mathbf{B}_{1}^{*} \mathbf{C}_1 \mathbf{V}_{1} \nonumber \\
  & = ( \mathbf{K}^{-1}+ \sum_{j=2}^{L}\mathbf{B}_{j}^{*})^{-1} \mathbf{B}_{1}^{*} ( \mathbf{K}^{-1}+ \sum_{j=1}^{L}\mathbf{B}_{j}^{*})^{-1} \mathbf{V}_{1} \nonumber \\
  & = ( \mathbf{K}^{-1}+ \sum_{j=2}^{L}\mathbf{B}_{j}^{*})^{-1} \mathbf{V}_{1} - ( \mathbf{K}^{-1}+ \sum_{j=1}^{L}\mathbf{B}_{j}^{*})^{-1} \mathbf{V}_{1} \nonumber \\
 & =\mathbf{C}_2 \mathbf{V}_{1} -\mathbf{C}_1 \mathbf{V}_{1},
\end{align}
which implies that
\begin{equation}\label{pass}
\mathbf{V}_{1}^{T}\mathbf{C}_1 \mathbf{V}_{1} = \mathbf{V}_{1}^{T}\mathbf{C}_2 \mathbf{V}_{1}.
\end{equation}
In view of \eqref{pass},  $\bm{e}_{n_{1}+1},\bm{e}_{n_{1}+2}, \ldots, \bm{e}_{m}$ are also the eigenvectors of matrix
$\mathbf{C}_2=( \mathbf{K}^{-1}+ \sum_{j=2}^{L}\mathbf{B}_{j}^{*})^{-1} $ with the eigenvalues being  $d_{n_{1}+1}, d_{n_{1}+2}, \ldots, d_{m}$. On the other hand, we can conclude that
\begin{equation}\label{eqn:cc-1}
\mathbf{U}_{1}^{T} \mathbf{C}_2  \mathbf{V}_{1} = \mathbf{U}_{1}^{T}\mathbf{C}_1  \mathbf{V}_{1} = 0.
\end{equation}\par
Subtracting \eqref{KKT1} from the first equation  in KKT conditions \eqref{KKT2} and invoking  \eqref{eqn:SD-1} gives
\begin{equation}\label{KKT_2}
\mathbf{\Delta}_{2}= \frac{\mu_{1}-\mu_{2}}{2} \mathbf{C}_2 +\mathbf{\Delta}_{1}.
\end{equation}
Combining equations \eqref{eqn:cc-1} and \eqref{KKT_2} with $\mathbf{U}_{1}^{T} \mathbf{\Delta}_{1}\mathbf{V}_{1} =0$, we see
\begin{equation}
\mathbf{U}_{1}^{T} \mathbf{\Delta}_{2}\mathbf{V}_{1} =0.
\end{equation}
Thus we may give matrix $ \mathbf{\Delta}_{2}$ the following spectral representation:
\begin{equation}
\mathbf{\Delta}_{2} = \mathbf{U}_{1}\mathbf{U}_{1}^{T} \mathbf{\Delta}_{2}\mathbf{U}_{1}\mathbf{U}_{1}^{T} + \mathbf{V}_{1}\mathbf{V}_{1}^{T} \mathbf{\Delta}_{2}\mathbf{V}_{1}\mathbf{V}_{1}^{T}.
\end{equation}
\smallskip
\par
From equation \eqref{KKT_2}, we have $\mathbf{\Delta}_{2} \succ \mathbf{\Delta}_{1} $ and consequently
$$\mathbf{U}_{1}^{T} \mathbf{\Delta}_{2}\mathbf{U}_{1} \succ \mathbf{U}_{1}^{T} \mathbf{\Delta}_{1}\mathbf{U}_{1}  \succ 0.$$
On the other hand,
\begin{align}
& \mathbf{V}_{1}\mathbf{V}_{1}^{T} \mathbf{\Delta}_{2}\mathbf{V}_{1}\mathbf{V}_{1}^{T}  \nonumber \\
 = &\frac{\mu_{1}-\mu_{2}}{2} \mathbf{V}_{1}\mathbf{V}_{1}^{T}\mathbf{C}_2 \mathbf{V}_{1}\mathbf{V}_{1}^{T}  + \mathbf{V}_{1}\mathbf{V}_{1}^{T} \mathbf{\Delta}_{1}\mathbf{V}_{1}\mathbf{V}_{1}^{T} \nonumber \\
= & \sum_{n=n_{1}+1}^{m} \frac{\mu_{1}-\mu_{2}}{2} d_{n} \bm{e}_{n} \bm{e}_{n}^{T} +\left (\lambda d_{n}^{2} - \frac{\mu_{1}}{2} d_{n}\right ) \bm{e}_{n} \bm{e}_{n}^{T} \nonumber \\
= & \sum_{n=n_{1}+1}^{m} \left (\lambda d_{n}^{2}- \mu_{2}\frac{d_{n}}{2}\right )\bm{e}_{n} \bm{e}_{n}^{T}.
\end{align}
\par
\smallskip
\noindent Now we are at the same situation as treating equation \eqref{eqn:add-1}, and correspondingly the refined spectral representation of matrix ${\mathbf\Delta}_2$ can be obtained through a procedure similar to that for ${\mathbf\Delta}_1$. Here we may   divide the subspace spanned by the column vector of $\mathbf{V}_{1}$ into two orthogonal subspaces, according to the
sign of ${\mathbf\Delta}_{2}$'s  eigenvalues $\lambda d_{n}^{2} - \mu_{2} d_{n} /2$, $n=n_{1}+1,n_{1}+2, \ldots, m$. Specifically, we partition the matrix $\mathbf{V}_1$
into a $m \times (n_{2}-n_{1})$  matrix $\mathbf{W}_{1} \triangleq  \left( \bm{e}_{n_{1}+1}, \bm{e}_{n_{1}+2}, \ldots, \bm{e}_{n_{2}}\right)$ and a $m \times (m-n_{2})$
  matrix $\mathbf{V}_{2} \triangleq  \left( \bm{e}_{n_{2}+1}, \bm{e}_{n_{2}+2}, \ldots, \bm{e}_m\right)$, in which $n_2$ represents the critical number such that
$$\lambda d_{n}^{2} - \mu_{2} d_{n} /2>0,\ \   n_1<n\leq n_2$$
$$\lambda d_{n}^{2} - \mu_{2} d_{n} /2\le 0,\ \   n_2<n\leq m.$$
On the other hand, combining  $\mathbf{U}_1$ and $\mathbf{W}_1$ will form a new $m \times n_{2}$ matrix $\mathbf{U}_2 \triangleq  \left( \bm{e}_{1}, \bm{e}_{2}, \ldots, \bm{e}_{n_{2}}\right)$.
%
%, corresponding to the positive eigenvalues $\lambda d_{n}^{2} - \mu_{2} d_{n} /2$ () of $\mathbf{V}_{1}$. And define the  $m \times (m-n_{2})$
%  matrix $\mathbf{V}_{2} \triangleq  \left( \bm{e}_{n_{2}+1}, \bm{e}_{n_{2}+2}, \ldots, \bm{e}_{n_{2}}\right)$, corresponding to the non-negative eigenvalues .
% in which the eigenvalues $\lambda d_{n}^{2} - \mu_{2} d_{n} /2$ corresponding to eigenvectors $\bm{e}_{n}, n=n_{1}+1, n_{1}+2,\ldots,n_{2}$ are all $> 0$,
%  the eigenvalues $\lambda d_{n}^{2} - \mu_{2} d_{n} /2$ corresponding to eigenvectors $\bm{e}_{n}, n=n_{2}+1, n_{2}+2,\ldots,m$ are all $\leq 0$.
   It is straightforward to verify that
\begin{align}
&\mathbf{W}_{1}^{T} \mathbf{\Delta}_{2} \mathbf{W}_{1}\succ 0; \quad \mathbf{V}_{2}^{T} \mathbf{\Delta}_{2} \mathbf{V}_{2}\preceq 0; \quad \mathbf{W}_{1}^{T} \mathbf{\Delta}_{2} \mathbf{V}_{2}=0; \\
&\mathbf{W}_{1}^{T}\mathbf{C}_2 \mathbf {W}_{1} \succeq 0;\quad
 \mathbf{V}_{2}^{T} \mathbf{C}_2 \mathbf {V}_{2} \succeq 0;\quad
 \mathbf{W}_{1}^{T}\mathbf{C}_2 \mathbf {V}_{2} = 0.
\end{align}
We can further refine the spectral decomposition form of $\mathbf{\Delta}_{2}$ and $\mathbf{C}_2$:
\begin{align}
&\mathbf{\Delta}_2= \mathbf{U}_{1}\mathbf{U}_{1}^{T}\mathbf{\Delta}_2 \mathbf{U}_{1}\mathbf{U}_{1}^{T} + \mathbf{V}_{1}\mathbf{V}_{1}^{T} \mathbf{\Delta}_2 \mathbf{V}_{1}\mathbf{V}_{1}^{T} \nonumber \\
=& \mathbf{U}_{1}\mathbf{U}_{1}^{T}\mathbf{\Delta}_2 \mathbf{U}_{1}\mathbf{U}_{1}^{T}+ \mathbf{W}_{1}\mathbf{W}_{1}^{T}\mathbf{\Delta}_2 \mathbf{W}_{1}\mathbf{W}_{1}^{T} \nonumber \\
& \quad + \mathbf{V}_{2}\mathbf{V}_{2}^{T}\mathbf{\Delta}_2 \mathbf{V}_{2}\mathbf{V}_{2}^{T}
\end{align}
\begin{align}
&\mathbf{C}_{2}
 = \mathbf{U}_{2}\mathbf{U}_{2}^{T} \mathbf{C}_{2}  \mathbf{U}_{2}\mathbf{U}_{2}^{T} + \mathbf{V}_{2}\mathbf{V}_{2}^{T} \mathbf{C}_{2}  \mathbf{V}_{2}\mathbf{V}_{2}^{T} \nonumber \\
=& \mathbf{U}_{1}\mathbf{U}_{1}^{T} \mathbf{C}_{2}  \mathbf{U}_{1}\mathbf{U}_{1}^{T} + \mathbf{W}_{1}\mathbf{W}_{1}^{T} \mathbf{C}_{2}  \mathbf{W}_{1}\mathbf{W}_{1}^{T}+\mathbf{V}_{2}\mathbf{V}_{2}^{T} \mathbf{C}_{2}  \mathbf{V}_{2}\mathbf{V}_{2}^{T}
\end{align}

\par
\smallskip

Following the similar steps as in the derivation of \eqref{eqn:b1}, we obtain
\begin{equation}
\mathbf{B}_{2}^{*}\mathbf{V}_{2}= 0.
\end{equation}
\begin{equation}
\mathbf{V}_{2}^{T}\mathbf{C}_2\mathbf{V}_{2}
=  \mathbf{V}_{2}^{T}\mathbf{C}_3 \mathbf{V}_{2}.
\end{equation}
%\begin{align}
%&\mathbf{V}_{2}^{T}( \mathbf{K}^{-1}+ \sum_{j=1}^{L}\mathbf{B}_{j}^{*})^{-1} \mathbf{V}_{2}
%=\mathbf{V}_{2}^{T}( \mathbf{K}^{-1}+ \sum_{j=2}^{L}\mathbf{B}_{j}^{*})^{-1} \mathbf{V}_{2} \nonumber \\
%= & \mathbf{V}_{2}^{T}( \mathbf{K}^{-1}+ \sum_{j=3}^{L}\mathbf{B}_{j}^{*})^{-1} \mathbf{V}_{2}.
%\end{align}
\par
Repeating this procedure $L$ times yields the following theorem.

\begin{theorem} \label{spectral}
In $\mathbb{R}^{m}$, there exist three sets of column orthogonal matrices\footnote{ One $m \times n $ dimensional $ (m \geq n )$ matrix $\mathbf{A}$ is called column orthogonal iff $\mathbf{A}^{T}\mathbf{A} = \mathbf{I}$.}
:
$\{ \mathbf{U}_{1}, \mathbf{U}_{2},\ldots, \mathbf{U}_{L}\}$, $\{ \mathbf{V}_{1}, \mathbf{V}_{2},\ldots, \mathbf{V}_{L}\}$, $\{ \mathbf{W}_{1}, \mathbf{W}_{2},\ldots, \mathbf{W}_{L-1}\}$,
 such that the following properties hold:
\begin{enumerate}
%\item \label{con0}[Structure of Matrix]
%\begin{equation}
%\mathbf{U}_{i+1}= \mathbf{U}_{i} \oplus \mathbf{W}_{i}, \; \mathbf{V}_{i} = \mathbf{V}_{i+1} \oplus \mathbf{W}_{i}, \; i=1,2,\ldots,L-1.
%\end{equation}
\item \label{con1}[Spectrum of $\mathbf{C}_{i}$]
%\begin{align}
% &\mathbf{C}_1 = \mathbf{U}_{1} \mathbf{U}_{1}^{T} \mathbf{C}_1 \mathbf {U}_{1} \mathbf{U}_{1}^{T}   + \mathbf{V}_{1}\mathbf{V}_{1}^{T} \mathbf{C}_1\mathbf {V}_{1}\mathbf{V}_{1}^{T},
%\end{align}
%
%\begin{align}
%\mathbf{C}_{i} & = \mathbf{U}_{i}\mathbf{U}_{i}^{T} \mathbf{C}_{i}  \mathbf{U}_{i}\mathbf{U}_{i}^{T}  + \mathbf{V}_{i}\mathbf{V}_{i}^{T} \mathbf{C}_{i} \mathbf{V}_{i}\mathbf{V}_{i}^{T} \nonumber \\
%& = \mathbf{U}_{i-1}\mathbf{U}_{i-1}^{T} \mathbf{C}_{i}  \mathbf{U}_{i-1}\mathbf{U}_{i-1}^{T}
%\nonumber \\
%& \quad+ \mathbf{W}_{i-1}\mathbf{W}_{i-1}^{T} \mathbf{C}_{i}  \mathbf{W}_{i-1}\mathbf{W}_{i-1}^{T}+\mathbf{V}_{i}\mathbf{V}_{i}^{T} \mathbf{C}_{i}  \mathbf{V}_{i}\mathbf{V}_{i}^{T}, \nonumber \\
%&  \qquad \qquad \qquad \qquad \qquad \qquad \qquad  \qquad  i=2,\ldots, L.
%\end{align}
\begin{align}
\mathbf{C}_{i} & = \mathbf{U}_{i}\mathbf{U}_{i}^{T} \mathbf{C}_{i}  \mathbf{U}_{i}\mathbf{U}_{i}^{T}  + \mathbf{V}_{i}\mathbf{V}_{i}^{T} \mathbf{C}_{i} \mathbf{V}_{i}\mathbf{V}_{i}^{T},  \; i=1,\ldots, L.
\end{align}
%\item \label{con2}
%\begin{align}
%&( \mathbf{K}^{-1}+ \sum_{j=1}^{L}\mathbf{B}_{j}^{*})^{-2} \nonumber \\
%= &  \mathbf{U}_{1}\mathbf{U}_{1}^{T}( \mathbf{K}^{-1}+ \sum_{j=1}^{L}\mathbf{B}_{j}^{*})^{-2}  \mathbf{U}_{1}  \mathbf{U}_{1}^{T}  \nonumber \\
%& \quad +  \mathbf{W}_{1}\mathbf{W}_{1}^{T}( \mathbf{K}^{-1}+ \sum_{j=1}^{L}\mathbf{B}_{j}^{*})^{-2}  \mathbf{W}_{1}  \mathbf{W}_{1}^{T}  \nonumber \\
%& \quad +  \cdots  \nonumber \\
%& \quad +  \mathbf{W}_{L-1}\mathbf{W}_{L-1}^{T}( \mathbf{K}^{-1}+ \sum_{j=1}^{L}\mathbf{B}_{j}^{*})^{-2}  \mathbf{W}_{L-1}  \mathbf{W}_{L-1}^{T}  \nonumber \\
%& \quad +  \mathbf{V}_{L}\mathbf{V}_{L}^{T}( \mathbf{K}^{-1}+ \sum_{j=1}^{L}\mathbf{B}_{j}^{*})^{-2}  \mathbf{V}_{L}  \mathbf{V}_{L}^{T}  \nonumber \\
%= & \sum_{n=1}^{m} d_{n}^{2} \bm{e}_{n}\bm{ e} _{n}^{T},
%\end{align}
%where $d_{n}, n=1,\cdots,m $ are the eigenvalues of matrix $( \mathbf{K}^{-1}+ \sum_{j=1}^{L}\mathbf{B}_{j}^{*})^{-1}$ corresponding to eigenvectors $ \bm{e}_{1}, \bm{e}_{2}, \ldots, \bm{e}_{m}$.

\item \label{con3}[Spectrum of $\mathbf{\Delta}_{i}$]
\begin{equation}
\mathbf{\Delta}_{1} = \mathbf{U}_{1}\mathbf{U}_{1}^{T} \mathbf{\Delta}_{1}\mathbf{U}_{1}\mathbf{U}_{1}^{T} + \mathbf{V}_{1}\mathbf{V}_{1}^{T} \mathbf{\Delta}_{1}\mathbf{V}_{1}\mathbf{V}_{1}^{T},
\end{equation}
\begin{align}
\mathbf{\Delta}_{i+1}
 & = \mathbf{U}_{i}\mathbf{U}_{i}^{T} \mathbf{\Delta}_{i+1}  \mathbf{U}_{i}\mathbf{U}_{i}^{T} + \mathbf{V}_{i}\mathbf{V}_{i}^{T} \mathbf{\Delta}_{i+1}  \mathbf{V}_{i}\mathbf{V}_{i}^{T} \nonumber \\
&= \mathbf{U}_{i}\mathbf{U}_{i}^{T} \mathbf{\Delta}_{i+1} \mathbf{U}_{i}\mathbf{U}_{i}^{T} + \mathbf{W}_{i}\mathbf{W}_{i}^{T}\mathbf{\Delta}_{i+1} \mathbf{W}_{i}\mathbf{W}_{i}^{T} \nonumber\\
 & \qquad \qquad+ \mathbf{V}_{i+1}\mathbf{V}_{i+1}^{T}\mathbf{\Delta}_{i+1} \mathbf{V}_{i+1}\mathbf{V}_{i+1}^{T}, \nonumber \\
& \qquad \qquad \qquad \qquad \qquad \qquad \qquad  i=1,\ldots, L-1.
\end{align}
\item \label{con4} [Positive/Negative definiteness]
\begin{align}
& \mathbf{U}_{i}^{T} \mathbf{\Delta}_{i} \mathbf{U}_{i} \succ 0, \qquad i=1, \ldots,L; \nonumber \\
& \mathbf{W}_{i}^{T} \mathbf{\Delta}_{i+1} \mathbf{W}_{i} \succ 0,  \qquad i=1, \ldots,L-1; \nonumber \\
& \mathbf{V}_{i}^{T}\mathbf{\Delta}_{i} \mathbf{V}_{i} \preceq 0,\qquad i=1, \ldots,L.
\end{align}
\item \label{con5} [Orthogonality]
For any $1 \leq i \leq L $,
\begin{equation}
\mathbf{B}_{i}^{*}\mathbf{V}_{i}=0.
\end{equation}

%\item \label {con6} [Projection equivalence]
%For any $1 \leq i \leq L-1$,
%\begin{align}
%\mathbf{V}_{i}^{T}\mathbf{C}_{i} \mathbf{V}_{i}
%=  \mathbf{V}_{i}^{T}\mathbf{C}_{i+1}  \mathbf{V}_{i}.
%\end{align}
\end{enumerate}
\end{theorem}
\smallskip
\par
%Via dividing $\mathbb{R}^{m}$ into several subspaces according to the eigenvalues of $\mathbf{C}_{1} =( \mathbf{K}^{-1}+ \sum_{j=1}^{L}\mathbf{B}_{j}^{*})^{-1}$, we have the
%following corollary on the spectral decomposition of $( \mathbf{K}^{-1}+ \sum_{j=1}^{L}\mathbf{B}_{j}^{*})^{-2}$.
%
%\begin{corollary}\label{fin}
%Let $d_{n}, n=1,\cdots,m $ be the eigenvalues of matrix $\mathbf{C}_{1} = ( \mathbf{K}^{-1}+ \sum_{j=1}^{L}\mathbf{B}_{j}^{*})^{-1}$ corresponding
%to eigenvectors $ \bm{e}_{1}, \bm{e}_{2}, \ldots, \bm{e}_{m}$.
%Denote $\mathbf{U}_{1}  \triangleq  \left( \bm{e}_{1}, \bm{e}_{2}, \ldots, \bm{e}_{n_{1}}\right)$; $\mathbf{W}_{i}  \triangleq  \left( \bm{e}_{n_{i}+1}, \bm{e}_{n_{i}+2}, \ldots, \bm{e}_{n_{i+1}}\right), i=1,L-1$;
%$\mathbf{V}_{L}  \triangleq  \left( \bm{e}_{n_{L}+1}, \bm{e}_{n_{L}+2}, \ldots, \bm{e}_{m}\right)$, where $1 \leq n_{1} \leq n_{2} \leq \cdots \leq n_{L} \leq m $. Then the following matrix identities hold:
%\begin{align}
%&  \mathbf{U}_{1}\mathbf{U}_{1}^T\mathbf{C}_{1}^2  \mathbf{U}_{1}  \mathbf{U}_{1}^{T} = \sum_{n=1}^{n_{1}} d_{n}^{2} \bm{e}_{n}\bm{ e} _{n}^{T};\label{fineq1} \\
%&  \mathbf{W}_{i}\mathbf{W}_{i}^{T}\mathbf{C}_{1}^2   \mathbf{W}_{i}  \mathbf{W}_{i}^{T} = \sum_{n=n_{i}+1}^{n_{i+1}} d_{n}^{2} \bm{e}_{n}\bm{ e} _{n}^{T}, \;\;\; i=1,\ldots, L-1; \label{fineq2} \\
%& \mathbf{V}_{L}\mathbf{V}_{L}^{T}\mathbf{C}_{1}^2  \mathbf{V}_{L}  \mathbf{V}_{L}^{T}
%=  \sum_{n=n_{L}+1}^{m} d_{n}^{2} \bm{e}_{n}\bm{ e} _{n}^{T}. \label{fineq3}
%\end{align}
%\end{corollary}

\section{Converse}

In this section we establish a new extremal inequality, which is further leveraged to give a complete characterization of the rate region of the vector Gaussian CEO problem with the trace distortion constraint. However, it appears difficult to give a direct proof of this extremal inequality using the perturbation method. To overcome this difficulty, we project the mean square error matrix of the Berger-Tung scheme into its eigenspaces,
  and estimate  each term of the extremal inequality in its respective subspace. This approach is partly inspired by the work of Rahman and Wagner on the vector Gaussian one-help-one problem \cite{vector}. \par

\subsection{Extremal Inequality }
\begin{theorem} \label{thm2}
Let $\mathbf{B}_{1}^{*}, \ldots, \mathbf{B}_{L}^{*}$ be the optimal solution of $R^{BT}(d)$.
For any random variables $(M_{1}, \ldots, M_{L}, Q)$ jointly distributed with $(\mathbf{X},\mathbf{Y}_1,\ldots,\mathbf{Y}_L)$ such that
\begin{align}
p(\mathbf{x}, \mathbf{y}_{1}, \ldots, \mathbf{y}_{L}, m_{1}, \ldots, m_{L},q) \nonumber \\
=p(\mathbf{x}) p(q) \prod_{i=1}^{L}p(\mathbf{y}_{i}|\mathbf{x})p(m_{i}|\mathbf{y}_{i},q),
\end{align}
and
\begin{align}
& \sum_{n=1}^{m} \mathbb{E} \left[ (\mathbf{x}_{n}- \mathbb{E}\left[\mathbf{x}_{n}| M_{1}, \ldots, M_{L} \right])^{2}\right]   \nonumber \\
= & \tr \left\{ \cov   (\mathbf{X} | M_{1}, \ldots, M_{L}) \right \} \nonumber \\
\leq & d,
\end{align}
we have
\begin{align}  \label{exinqq}
&\sum_{i=1}^{L-1} (\mu_{i}- \mu_{i+1}) h(\mathbf{X}| M_{i+1}, \ldots, M_{L}) \nonumber \\
&\quad - \mu_{1} h(\mathbf{X}| M_{1}, \ldots, M_{L}) - \sum_{i=1}^{L} \mu_{i} h(\mathbf{Y}_{i} | \mathbf{X}, M_{i},Q) \nonumber \\
\geq & \sum_{i=1}^{L-1}\frac{\mu_{i} - \mu_{i+1}}{2} \log |(2\pi e) \mathbf{C}_{i+1}|  - \frac{\mu_{1}}{2} \log |(2\pi e)\mathbf{C}_{1}| \nonumber \\
& \quad \quad - \sum_{i=1}^{L}\frac{\mu_{i}}{2} \log |(2\pi e)( \mathbf{\Sigma}_{i} -  \mathbf{\Sigma}_{i} \mathbf{B}_{i}^{*} \mathbf{\Sigma}_{i} )|.
\end{align}
\end{theorem}
\smallskip

Note that
\begin{align}
 & h(\left[\mathbf{U}_{i},  \ (\mathbf{\Sigma}_{i}^{-1}- \mathbf{B}_{i}^{*})^{-1} \mathbf{V}_{i} \right]^{T}\mathbf{\Sigma}_{i}^{-1}\mathbf{Y}_{i} | \mathbf{X}, M_{i},Q) \nonumber \\
\leq &  h(\mathbf{U}_{i}^{T}\mathbf{\Sigma}_{i}^{-1}\mathbf{Y}_{i} | \mathbf{X}, M_{i},Q) \nonumber \\
& +  h(\mathbf{V}_{i}^{T}(\mathbf{\Sigma}_{i}^{-1}- \mathbf{B}_{i}^{*})^{-1}\mathbf{\Sigma}_{i}^{-1}\mathbf{Y}_{i} | \mathbf{X}, M_{i},Q).\nonumber
\end{align}
On the other hand, following by the matrix equality, 
\begin{align}
    & (2 \pi e)
   \begin{pmatrix}
  \mathbf{U}_{i}^{T}\\
   \mathbf{V}_{i}^{T} (\mathbf{\Sigma}_{i}^{-1} - \mathbf{B}_{i}^{*})^{-1}
   \end{pmatrix} \mathbf{\Sigma}_{i}^{-1} (\mathbf{\Sigma}_{i}-\mathbf{\Sigma}_{i}\mathbf{B}_{i}^{*}\mathbf{\Sigma}_{i})\mathbf{\Sigma}_{i}^{-1}
   \nonumber \\
   & \qquad \cdot  \begin{pmatrix}
   \mathbf{U}_{i} &
   (\mathbf{\Sigma}_{i}^{-1} - \mathbf{B}_{i}^{*})^{-1}\mathbf{V}_{i}
   \end{pmatrix} \nonumber \\
   = &
   \begin{pmatrix}
   (2 \pi e) \mathbf{U}_{i}^{T}(\mathbf{\Sigma}_{i}^{-1} - \mathbf{B}_{i}^{*})\mathbf{U}_{i} & 0 \\
   0 & (2 \pi e) \mathbf{V}_{i}^{T}(\mathbf{\Sigma}_{i}^{-1} - \mathbf{B}_{i}^{*})^{-1}\mathbf{V}_{i}
   \end{pmatrix}. \nonumber
   \end{align}

 By Taking logarithm for the determinant of matrix to both sides,
we have
\begin{align}
&\frac{1}{2}\log |(2\pi e)( \mathbf{\Sigma}_{i} -  \mathbf{\Sigma}_{i} \mathbf{B}_{i}^{*} \mathbf{\Sigma}_{i} )| +  \log |\mathbf{\Sigma}^{-1}|  \nonumber \\
& \quad + \log |\left[\mathbf{U}_{i},  \ (\mathbf{\Sigma}_{i}^{-1}- \mathbf{B}_{i}^{*})^{-1} \mathbf{V}_{i} \right]|\nonumber \\
= & \frac{1}{2}\log |(2\pi e)\mathbf{U}_{i}^{T}(\mathbf{\Sigma}_{i}^{-1}- \mathbf{B}_{i}^{*})\mathbf{U}_{i}| \nonumber \\
& \quad +\frac{1}{2}\log |(2\pi e)\mathbf{V}_{i}^{T}(\mathbf{\Sigma}_{i}^{-1}- \mathbf{B}_{i}^{*})^{-1}\mathbf{V}_{i}|. \nonumber
\end{align}
Therefore, it suffices to prove
\begin{align}  \label{exinq}
&\sum_{i=1}^{L-1} (\mu_{i}- \mu_{i+1}) h(\mathbf{X}| M_{i+1}, \ldots, M_{L}) \nonumber \\
&\quad - \mu_{1} h(\mathbf{X}| M_{1}, \ldots, M_{L}) \nonumber \\
&\quad - \sum_{i=1}^{L} \mu_{i} h(\mathbf{U}_{i}^{T}\mathbf{\Sigma}_{i}^{-1}\mathbf{Y}_{i} | \mathbf{X}, M_{i},Q) \nonumber \\
&\quad - \sum_{i=1}^{L} h(\mathbf{V}_{i}^{T}(\mathbf{\Sigma}_{i}^{-1}- \mathbf{B}_{i}^{*})^{-1}\mathbf{\Sigma}_{i}^{-1}\mathbf{Y}_{i} | \mathbf{X}, M_{i},Q) \nonumber \\
\geq & \sum_{i=1}^{L-1}\frac{\mu_{i} - \mu_{i+1}}{2} \log |(2\pi e) \mathbf{C}_{i+1}|- \frac{\mu_{1}}{2} \log |(2\pi e)\mathbf{C}_{1}|\nonumber \\
& -\sum_{i=1}^{L} \frac{\mu_{i}}{2}\log |(2\pi e)\mathbf{U}_{i}^{T}(\mathbf{\Sigma}_{i}^{-1}- \mathbf{B}_{i}^{*})\mathbf{U}_{i}| \nonumber \\
& -\sum_{i=1}^{L} \frac{\mu_{i}}{2}\log |(2\pi e)\mathbf{V}_{i}^{T}(\mathbf{\Sigma}_{i}^{-1}- \mathbf{B}_{i}^{*})^{-1}\mathbf{V}_{i}|
\end{align}

\smallskip

\begin{figure*}[!t]
% ensure that we have normalsize text
\normalsize
% Store the current equation number.
\setcounter{MYtempeqncnt}{\value{equation}}
\addtocounter{equation}{6}
% Set the equation number to one less than the one
% desired for the first equation here.
% The value here will have to changed if equations
% are added or removed prior to the place these
% equations are referenced in the main text.
\setcounter{equation}{\value{equation}}
 \begin{align}\label{tmm}
& \frac{d}{d\gamma} h(\mathbf{U}_{i}^{T}\mathbf{\Sigma}_{i}^{-1}\mathbf{Y}_{i,\gamma}|\mathbf{X}, M_{i},Q) \nonumber \\
= & \frac{d}{d\gamma} \Big\{h(\sqrt{\frac{1-\gamma}{\gamma}} \mathbf{U}_{i}^{T}\mathbf{\Sigma}_{i}^{-1}\mathbf{Y}_{i} +  \mathbf{U}_{i}^{T}\mathbf{\Sigma}_{i}^{-1}\mathbf{N}_{i}^{G} | \mathbf{X}, M_{i},Q) + n_{i}  \log \gamma \Big\} \nonumber \\
= & \frac{1}{2} \tr \Big\{ \frac{1}{\gamma}\mathbf{I} - \frac{1}{\gamma^{2}}\big(\mathbf{U}_{i}^{T}( \mathbf{\Sigma}_{i}^{-1} -  \mathbf{B}_{i}^{*})\mathbf{U}_{i}\big)^{-1} \cov (\mathbf{U}_{i}^{T}\mathbf{\Sigma}_{i}^{-1}\mathbf{Y}_{i} | \sqrt{\frac{1-\gamma}{\gamma}} \mathbf{U}_{i}^{T}\mathbf{\Sigma}_{i}^{-1}\mathbf{Y}_{i} +  \mathbf{U}_{i}^{T}\mathbf{\Sigma}_{i}^{-1}\mathbf{N}_{i}^{G},  \mathbf{X}, M_{i},Q)\Big \} \nonumber \\
= & \frac{1}{2} \tr \Big\{ \frac{1}{\gamma}\mathbf{I} - \frac{1}{\gamma^{2}(1-\gamma)}\big(\mathbf{U}_{i}^{T}( \mathbf{\Sigma}_{i}^{-1} -  \mathbf{B}_{i}^{*})\mathbf{U}_{i}\big)^{-1} \cov ( \sqrt{1-\gamma} \mathbf{U}_{i}^{T}\mathbf{\Sigma}_{i}^{-1}\mathbf{Y}_{i} | \sqrt{1-\gamma} \mathbf{U}_{i}^{T}\mathbf{\Sigma}_{i}^{-1}\mathbf{Y}_{i} +  \sqrt {\gamma }\mathbf{U}_{i}^{T}\mathbf{\Sigma}_{i}^{-1}\mathbf{N}_{i}^{G},  \mathbf{X}, M_{i},Q)\Big \} \nonumber \\
\overset{(a)}{\geq}  & \frac{1}{2} \tr \Big\{ \frac{1}{\gamma}\mathbf{I} - \frac{1}{\gamma^{2}(1-\gamma)}\big(\mathbf{U}_{i}^{T}( \mathbf{\Sigma}_{i}^{-1} -  \mathbf{B}_{i}^{*} )\mathbf{U}_{i}\big)^{-1}  \big( \gamma^{2}\cov (\sqrt{1-\gamma}\mathbf{U}_{i}^{T}\mathbf{\Sigma}_{i}^{-1}\mathbf{Y}_{i} | \mathbf{X}, M_{i}, Q)  +   \gamma(1-\gamma)^{2} \mathbf{U}_{i}^{T}( \mathbf{\Sigma}_{i}^{-1} -  \mathbf{B}_{i}^{*} )\mathbf{U}_{i}\big)\Big \}.\end{align}
% Restore the current equation number.
\setcounter{equation}{\value{MYtempeqncnt}}
% IEEE uses as a separator
\hrulefill
% The spacer can be tweaked to stop underfull vboxes.
\vspace*{4pt}
\end{figure*}

\begin{figure*}[!t]
% ensure that we have normalsize text
\normalsize
% Store the current equation number.
\setcounter{MYtempeqncnt}{\value{equation}}
\addtocounter{equation}{7}
% Set the equation number to one less than the one
% desired for the first equation here.
% The value here will have to changed if equations
% are added or removed prior to the place these
% equations are referenced in the main text.
\setcounter{equation}{\value{equation}}
 \begin{align}\label{tmm2}
& \frac{d}{d\gamma} h(\mathbf{V}_{i}^{T}(\mathbf{\Sigma}_{i}^{-1} -\mathbf{B}_{i}^{*})^{-1}\mathbf{\Sigma}_{i}^{-1}\mathbf{Y}_{i,\gamma}|\mathbf{X}, M_{i},Q) \nonumber \\
= & \frac{d}{d\gamma} \Big\{h(\mathbf{V}_{i}^{T}(\mathbf{\Sigma}_{i}^{-1} -\mathbf{B}_{i}^{*})^{-1}\mathbf{\Sigma}_{i}^{-1} \mathbf{Y}_{i}+    \sqrt{\frac{\gamma}{1-\gamma}}\mathbf{V}_{i}^{T}(\mathbf{\Sigma}_{i}^{-1} -\mathbf{B}_{i}^{*})^{-1}\mathbf{\Sigma}_{i}^{-1}\mathbf{N}_{i}^{G} | \mathbf{X}, M_{i},Q) + (n-n_{i})  \log \gamma \Big\} \nonumber \\
= & \frac{1}{2} \tr \Big\{\frac{1}{(1-\gamma)^{2}}J(\mathbf{V}_{i}^{T}(\mathbf{\Sigma}_{i}^{-1} -\mathbf{B}_{i}^{*})^{-1} \mathbf{\Sigma}_{i}^{-1}\mathbf{Y}_{i}+    \sqrt{\frac{\gamma}{1-\gamma}}\mathbf{V}_{i}^{T}(\mathbf{\Sigma}_{i}^{-1} -\mathbf{B}_{i}^{*})^{-1}\mathbf{\Sigma}_{i}^{-1}\mathbf{N}_{i}^{G} | \mathbf{X}, M_{i},Q)\mathbf{V}_{i}^{T}(\mathbf{\Sigma}_{i}^{-1} -\mathbf{B}_{i}^{*})^{-1}\mathbf{V}_{i}- \frac{1}{1-\gamma}\mathbf{I} \Big \} \nonumber \\
\overset{(a)}{\geq} & \frac{1}{2} \tr \Big\{\frac{1}{(1-\gamma)^{2}}J(\mathbf{V}_{i}^{T}(\mathbf{\Sigma}_{i}^{-1} -\mathbf{B}_{i}^{*})^{-1}\mathbf{\Sigma}_{i}^{-1} \mathbf{Y}_{i}+    \sqrt{\frac{\gamma}{1-\gamma}}\mathbf{V}_{i}^{T}(\mathbf{\Sigma}_{i}^{-1} -\mathbf{B}_{i}^{*})^{-1}\mathbf{\Sigma}_{i}^{-1}\mathbf{N}_{i}^{G} | \mathbf{X})\mathbf{V}_{i}^{T}(\mathbf{\Sigma}_{i}^{-1} -\mathbf{B}_{i}^{*})^{-1}\mathbf{V}_{i}- \frac{1}{1-\gamma}\mathbf{I} \Big \} \nonumber \\
\overset{(b)}{=} & \frac{1}{2(1-\gamma)} \tr \Big\{  \mathbf{V}_{i}^{T}  \big( (1-\gamma)(\mathbf{\Sigma}_{i}^{-1} -\mathbf{B}_{i}^{*}) \mathbf{\Sigma}_{i} (\mathbf{\Sigma}_{i}^{-1} -\mathbf{B}_{i}^{*}) + \gamma (\mathbf{\Sigma}_{i}^{-1} -\mathbf{B}_{i}^{*})\big)\mathbf{V}_{i}\mathbf{V}_{i}^{T}(\mathbf{\Sigma}_{i}^{-1} -\mathbf{B}_{i}^{*})^{-1}\mathbf{V}_{i} - \mathbf{I}\Big\} \nonumber \\
\overset{(c)}{=}&  \frac{1}{2(1-\gamma)} \tr \Big\{ \mathbf{V}_{i}^{T}(\mathbf{\Sigma}_{i}^{-1} -\mathbf{B}_{i}^{*})\mathbf{V}_{i} \mathbf{V}_{i}^{T}(\mathbf{\Sigma}_{i}^{-1} -\mathbf{B}_{i}^{*})^{-1}\mathbf{V}_{i}- \mathbf{V}_{i}^{T}\mathbf{V}_{i}\Big\}
\end{align}
% Restore the current equation number.
\setcounter{equation}{\value{MYtempeqncnt}}
% IEEE uses as a separator
\hrulefill
% The spacer can be tweaked to stop underfull vboxes.
\vspace*{4pt}
\end{figure*}

To the end of proving inequality \eqref{exinq}, we define $2L$ mutually independent zero mean Gaussian distributed random vectors
 $\mathbf{X}_{\{1,\ldots, L\}}^{G}, \mathbf{X}_{\{2,\ldots, L\}}^{G}, \ldots, \mathbf{X}_{\{L\}}^{G}$ and $\mathbf{N}_{1}^{G}, \mathbf{N}_{2}^{G}, \ldots, \mathbf{N}_{L}^{G}$,
which are independent of  $(\mathbf{X}, \mathbf{Y}_{1}, \ldots, \mathbf{Y}_{L}, M_{1}, \ldots, M_{L},Q)$. Here their distributions are
\begin{align}
 &\mathbf{X}_{\{i, \ldots, L\}}^{G} \sim \mathcal{N}\left(0,  ( \mathbf{K}^{-1}+ \mathbf{B}_{i}^{*} + \cdots + \mathbf{B}_{L}^{*})^{-1}\right), & i=1, \ldots, L;\nonumber \\
 &\mathbf{N}_{i}^{G} \sim  \mathcal{N}\left(0, ( \mathbf{\Sigma}_{i} -  \mathbf{\Sigma}_{i} \mathbf{B}_{i}^{*} \mathbf{\Sigma}_{i} )\right), & i=1, \ldots, L.\nonumber
\end{align}
Following \cite{exinq},\cite{my}, we use the {\it covariance preserved transform} proposed by Dembo \emph{et al.} in \cite{Dembo}. Specifically, for any $\gamma \in (0,1)$, define
\begin{align}\label{eq:dembo}
& \mathbf{X}_{i,\gamma} = \sqrt{1-\gamma} \mathbf{X} + \sqrt{\gamma} \mathbf{X}_{\{i,\ldots,L\}}^{G},      &i=1,\ldots,L; \nonumber \\
& \mathbf{Y}_{i,\gamma} = \sqrt{1-\gamma} \mathbf{Y}_{i} + \sqrt{\gamma} \mathbf{N}_{i}^{G},      &i=1,\ldots,L.
\end{align}
\smallskip
Consider the functional
\begin{align}
g(\gamma) = &\sum_{i=1}^{L-1} (\mu_{i}- \mu_{i+1}) h(\mathbf{X}_{i+1,\gamma}| M_{i+1}, \ldots, M_{L}) \nonumber \\
&\quad - \mu_{1} h(\mathbf{X}_{1,\gamma}| M_{1}, \ldots, M_{L}) \nonumber \\
&\quad - \sum_{i=1}^{L} \mu_{i} h(\mathbf{U}_{i}^{T}\mathbf{\Sigma_{i}}^{-1}\mathbf{Y}_{i,\gamma} | \mathbf{X}, M_{i},Q) \nonumber \\
&\quad - \sum_{i=1}^{L} \mu_{i} h(\mathbf{V}_{i}^{T}(\mathbf{\Sigma}_{i}^{-1}- \mathbf{B}_{i}^{*})^{-1}\mathbf{\Sigma}_{i}^{-1}{\mathbf{Y}}_{i,\gamma} | \mathbf{X}, M_{i},Q). \nonumber
\end{align}
The following lemma is needed for evaluating the derivative of $g(\gamma)$ with respect to $\gamma$.
%to lower bound the function $g(\gamma)$, we are interested in its derivative on perturbation variable $\gamma$, which will be derived in the next lemma.
\begin{lemma}\label{der}
For the afore-defined  $\mathbf{X}_{i,\gamma}$ and $\mathbf{Y}_{i,\gamma}$, we have
\begin{enumerate}
\item
\begin{align}
&2(1-\gamma) \frac{d}{d\gamma} h(\mathbf{X}_{i,\gamma}| M_{i}, \ldots, M_{L}) \nonumber \\
=& \tr \left\{ \mathbf{C}_{i} \left( J (\mathbf{X}_{i,\gamma}| M_{i}, \ldots, M_{L}) - \mathbf{C}_{i}^{-1} \right) \right\}
\end{align}
\item
\begin{align} \label{dep:2}
&2(1-\gamma) \frac{d}{d\gamma} h(\mathbf{U}_{i}^{T}\mathbf{\Sigma}_{i}^{-1}\mathbf{Y}_{i,\gamma}|\mathbf{X}, M_{i},Q) \nonumber \\
\geq & \tr \left\{\mathbf{U}_{i} \mathbf{U}_{i}^{T} -\mathbf{U}_{i} \mathbf{U}_{i}^{T}( \mathbf{\Sigma}_{i}^{-1} -  \mathbf{B}_{i}^{*} )^{-1}\mathbf{U}_{i} \mathbf{U}_{i}^{T} \right. \nonumber \\
 & \quad \left.\cdot \mathbf{\Sigma}_{i}^{-1}\cov (\mathbf{Y}_{i,\gamma}|\mathbf{X}, M_{i},Q) \mathbf{\Sigma}_{i}^{-1} \right\}
\end{align}
\item
\begin{align} \label{dep:3}
& 2(1-\gamma) \frac{d}{d\gamma} h(\mathbf{V}_{i}^{T}(\mathbf{\Sigma}_{i}^{-1}- \mathbf{B}_{i}^{*})^{-1}\mathbf{\Sigma}_{i}^{-1}{\mathbf{Y}}_{i,\gamma} | \mathbf{X}, M_{i},Q) \nonumber \\
\geq & \tr \left\{ \mathbf{V}_{i} \mathbf{V}_{i}^{T}(\mathbf{\Sigma}_{i}^{-1}- \mathbf{B}_{i}^{*})^{-1} \mathbf{V}_{i} \mathbf{V}_{i}^{T}(\mathbf{\Sigma}_{i}^{-1}- \mathbf{B}_{i}^{*})\right. \nonumber \\
& \quad \left.- \mathbf{V}_{i} \mathbf{V}_{i}^{T}\right\}
\end{align}

\end{enumerate}
\end{lemma}
\smallskip
\begin{proof}
\begin{enumerate}
\item
Using  de Bruijn's identity \eqref{lemma3} in Lemma \ref{lea3} and taking  $\gamma'=\gamma/(1-\gamma)$, we obtain
\begin{align}
& \frac{d}{d\gamma}  h(\mathbf{X}_{i,\gamma}| M_{i}, \ldots, M_{L}) \nonumber \\
= & \frac{d}{d\gamma} \Big\{h(\mathbf{X} + \sqrt{\frac{\gamma}{1-\gamma}} \mathbf{X}_{\{i,\ldots,L\}}^{G} | M_{i}, \ldots, M_{L}) \nonumber \\
{} & \quad + n \log (1-\gamma) \Big\} \nonumber \\
= & \frac{1}{2} \tr \Big\{ \frac{1}{(1-\gamma)^{2}}J(\mathbf{X} + \sqrt{\frac{\gamma}{1-\gamma}} \mathbf{X}_{\{i,\ldots,L\}}^{G} | M_{i}, \ldots, M_{L}) \nonumber \\
 & \quad \cdot \mathbf{C}_{i} - \frac{1}{1-\gamma}\mathbf{I} \Big\}.
\end{align}
Multiplying both sides with $2(1-\gamma)$ yields
\begin{align}
&2(1-\gamma) \frac{d}{d\gamma} h(\mathbf{X}_{i,\gamma}| M_{i}, \ldots, M_{L}) \nonumber \\
=& \tr \Big \{ J(\sqrt {1-\gamma}\mathbf{X} + \sqrt{\gamma} \mathbf{X}_{\{i,\ldots,L\}}^{G} | M_{i}, \ldots, M_{L}) %\nonumber \\
%& \quad \cdot
\mathbf{C}_{i} - \mathbf{I} \Big\} \nonumber \\
=& \tr \left\{  J (\mathbf{X}_{i,\gamma}| M_{i}, \ldots, M_{L})\mathbf{C}_{i}- \mathbf{I}  \right\} \nonumber \\
= & \tr \left\{  \mathbf{C}_{i} \left( J (\mathbf{X}_{i,\gamma}| M_{i}, \ldots, M_{L})-\mathbf{C}_{i}^{-1})\right)\right\}.
\end{align}

\item
Using the alternative form of de Bruijn's identity \eqref{corollary1} in Corollary \ref{cor1} and taking  $\gamma'=(1-\gamma)/\gamma$, we obtain inequality \eqref{tmm} at the top of next page.
\addtocounter{equation}{1}

In \eqref{tmm}, inequality (a) follows from Lemma \ref{lea5}. Multiplying  both sides of \eqref{tmm} $2(1-\gamma)$  gives

\begin{align}
&2(1-\gamma) \frac{d}{d\gamma}h(\mathbf{U}_{i}^{T}\mathbf{\Sigma}_{i}^{-1}\mathbf{Y}_{i,\gamma}|\mathbf{X}, M_{i},Q) \nonumber \\
\geq & \tr \Big( \mathbf{I}-\big( \mathbf{U}_{i}^{T}( \mathbf{\Sigma}_{i}^{-1} -  \mathbf{B}_{i}^{*} )\mathbf{U}_{i} \big)^{-1}\nonumber \\
  \cov& (\sqrt{1-\gamma} \mathbf{U}_{i}^{T}\mathbf{\Sigma}_{i}^{-1}\mathbf{Y}_{i}+ \sqrt{\gamma}\mathbf{U}_{i}^{T} \mathbf{\Sigma}_{i}^{-1}\mathbf{N}_{i}^{G} | \mathbf{X}, M_{i},Q) \Big\} \nonumber \\
\overset{(a)}{\geq} & \tr \left\{ \mathbf{U}_{i}^{T}\mathbf{U}_{i} - \mathbf{U}_{i}^{T}( \mathbf{\Sigma}_{i}^{-1} -   \mathbf{B}_{i}^{*} )^{-1}\mathbf{U}_{i} \right. \nonumber \\
 & \quad \left.\cdot \mathbf{U}_{i}^{T}\mathbf{\Sigma}_{i}^{-1}\cov (\mathbf{Y}_{i,\gamma}|\mathbf{X}, M_{i},Q)\mathbf{\Sigma}_{i}^{-1}\mathbf{U}_{i}  \right\} \nonumber \\
 = & \tr \left\{\mathbf{U}_{i} \mathbf{U}_{i}^{T} -\mathbf{U}_{i} \mathbf{U}_{i}^{T}( \mathbf{\Sigma}_{i}^{-1} -   \mathbf{B}_{i}^{*}  )^{-1}\mathbf{U}_{i} \mathbf{U}_{i}^{T} \right. \nonumber \\
 & \quad \left.\cdot \mathbf{\Sigma}_{i}^{-1}\cov (\mathbf{Y}_{i,\gamma}|\mathbf{X}, M_{i},Q)\mathbf{\Sigma}_{i}^{-1}  \right\}, \nonumber
\end{align}
where (a) follows from the simple fact that for any positive definite matrix $\mathbf{A}$ and column orthogonal matrix $\mathbf{P}$,
$$\left(\mathbf{P}^{T} \mathbf{A} \mathbf{P}\right)^{-1} \preceq \mathbf{P}^{T} \mathbf{A}^{-1} \mathbf{P}.$$

\item
Again using de Bruijn's identity \eqref{lemma3} in Lemma \ref{lea3} and taking  $\gamma'=\gamma/(1-\gamma)$, we obtain
inequality \eqref{tmm2} at the top of next page. \par

\addtocounter{equation}{1}
In \eqref{tmm2}, (a) follows from the data processing inequality of Fisher information matrix in Lemma \ref{lea6}; (b) is due to the fact that $(\mathbf{Y}_{i}, \mathbf{X}_{i})$ and $\mathbf{N}_{i}^{G}$ are independently  distributed Gaussians; (c) is due to $ \mathbf{B}_{i}^{*}\mathbf{V}_{i}=0$ (see Proposition \ref{con4} in Theorem \ref{spectral}). By multiplying both sides of \eqref{tmm2} with $2(1-\gamma)$, and switching the matrices in the trace operator, we obtain \eqref{dep:3} as desired.

\end{enumerate}
\end{proof}

Since
\begin{align} \label{ttt}
&\tr \left\{ \mathbf{V}_{i} \mathbf{V}_{i}^{T}(\mathbf{\Sigma}_{i}^{-1}- \mathbf{B}_{i}^{*})^{-1} \mathbf{V}_{i} \mathbf{V}_{i}^{T}(\mathbf{\Sigma}_{i}^{-1}- \mathbf{B}_{i}^{*})- \mathbf{V}_{i} \mathbf{V}_{i}^{T}\right\}\nonumber \\
= & \tr \left\{\mathbf{U}_{i} \mathbf{U}_{i}^{T}(\mathbf{\Sigma}_{i}^{-1}- \mathbf{B}_{i}^{*})^{-1} \mathbf{U}_{i} \mathbf{U}_{i}^{T}(\mathbf{\Sigma}_{i}^{-1}- \mathbf{B}_{i}^{*})- \mathbf{U}_{i} \mathbf{U}_{i}^{T}\right\},
\end{align}
it follows by \eqref{ttt} and Lemma \ref{der} that
\begin{align} \label{bound1}
& 2(1-\gamma)g'(\gamma) \nonumber \\
\leq & \sum_{i=1}^{L-1}\tr \Big\{  (\mu_{i}-\mu_{i+1})\mathbf{C}_{i+1} \nonumber \\
 & \qquad \cdot \Big(J (\mathbf{X}_{i+1,\gamma}| M_{i+1}, \ldots, M_{L}) -  \mathbf{C}_{i+1}^{-1}\Big) \Big\} \nonumber \\
 & -  \tr \Big\{  \mu_{1}\mathbf{C}_{1}   \cdot \Big( J (\mathbf{X}_{1,\gamma}| M_{1}, \ldots, M_{L})- \mathbf{C}_{1}^{-1} \Big)\Big\} \nonumber \\
  &- \sum_{i=1}^{L} \tr \Big\{ \mu_{i} \mathbf{U}_{i} \mathbf{U}_{i}^{T}( \mathbf{\Sigma}_{i}^{-1} - \mathbf{B}_{i}^{*}  )^{-1}\mathbf{U}_{i} \mathbf{U}_{i}^{T}\nonumber  \\
 &\qquad \cdot \Big(  ( \mathbf{\Sigma}_{i}^{-1} -  \mathbf{B}_{i}^{*} ) - \mathbf{\Sigma}_{i}^{-1} \cov (\mathbf{Y}_{i,\gamma}|\mathbf{X}, M_{i},Q) \mathbf{\Sigma}_{i}^{-1} \Big)\Big\}.
\end{align}
\par
Notice that when $\gamma=0$, $g(\gamma)$ equals to l.h.s. of extremal inequality \eqref{exinq}; when  $\gamma=1$, $g(\gamma)$ equals to r.h.s. of extremal inequality \eqref{exinq}.  We have the following theorem regarding the derivative of $g(\gamma)$ with respect to $\gamma$, and its proof is given in the next section.
\begin{theorem}\label{thm:33}
%Assume $g(\gamma)$ is defined as before and let $\gamma $ be the perturbation variable,
We have
\begin{equation} \label{thm:3}
2(1-\gamma) g'(\gamma) \leq 0.
\end{equation}
\end{theorem}
\smallskip
Note that \eqref{thm:3} implies the existence of  a monotonically decreasing path from $\gamma=0$ to $\gamma=1$, from which the desired extremal inequality follows immediately. %Thus the next  section is devoted to prove the r.h.s of the inequality \eqref{bound1} is upper bounded by $0$.

\begin{figure*}[!t]
% ensure that we have normalsize text
\normalsize
% Store the current equation number.
\setcounter{MYtempeqncnt}{\value{equation}}
\addtocounter{equation}{1}
% Set the equation number to one less than the one
% desired for the first equation here.
% The value here will have to changed if equations
% are added or removed prior to the place these
% equations are referenced in the main text.
\setcounter{equation}{\value{equation}}
\begin{subequations}
 \begin{align}
I_{1} = & \sum_{i=1}^{L-1}\tr \Big\{  2 \mathbf{U}_{i}\mathbf{U}_{i}^{T}(\mathbf{\Delta}_{i+1}-\mathbf{\Delta}_{i}) \mathbf{U}_{i}\mathbf{U}_{i}^{T}\Big(J (\mathbf{X}_{i+1,\gamma}| M_{i+1}, \ldots, M_{L}) - \mathbf{C}_{i+1}^{-1} \Big) \Big\} \nonumber \\
& + \tr \Big\{  2 \mathbf{U}_{1}\mathbf{U}_{1}^{T}\mathbf{\Delta}_{1} \mathbf{U}_{1}\mathbf{U}_{1}^{T}\Big( J (\mathbf{X}_{1,\gamma}| M_{1}, \ldots, M_{L})-\mathbf{C}_{1}^{-1} \Big)\Big\}; \label{sub:1} \\
 I_{2} = &  \sum_{i=1}^{L-1}\tr \Big\{  2 \mathbf{V}_{i}\mathbf{V}_{i}^{T}(\mathbf{\Delta}_{i+1}-\mathbf{\Delta}_{i}) \mathbf{V}_{i}\mathbf{V}_{i}^{T}\Big(J (\mathbf{X}_{i+1,\gamma}| M_{i+1}, \ldots, M_{L}) - \mathbf{C}_{i+1}^{-1}\Big) \Big\}\nonumber\\
 & +  \tr \Big\{  2 \mathbf{V}_{1}\mathbf{V}_{1}^{T}\mathbf{\Delta}_{1}\mathbf{V}_{1}\mathbf{V}_{1}^{T}\Big( J (\mathbf{X}_{1,\gamma}| M_{1}, \ldots, M_{L})-\mathbf{C}_{1}^{-1} \Big)\Big\}; \label{sub:2}\\
 I_{3} = & - \sum_{i=1}^{L} \tr \Big\{\mu_{i} \mathbf{U}_{i} \mathbf{U}_{i}^{T}( \mathbf{\Sigma}_{i}^{-1} - \mathbf{B}_{i}^{*}  )^{-1}\mathbf{U}_{i} \mathbf{U}_{i}^{T}\Big(  ( \mathbf{\Sigma}_{i}^{-1} -  \mathbf{B}_{i}^{*} ) - \mathbf{\Sigma}_{i}^{-1} \cov (\mathbf{Y}_{i,\gamma}|\mathbf{X}, M_{i},Q) \mathbf{\Sigma}_{i}^{-1} \Big)\Big\}. \label{sub:3} \\
 I_{4} = & - 2 \lambda\tr \Big\{  \mathbf{C}_{1}^2
 \Big( J (\mathbf{X}_{1,\gamma}| M_{1}, \ldots, M_{L})-\mathbf{C}_{1}^{-1} \Big)\Big\}.\label{sub:4}
 \end{align}
 \end{subequations}
% Restore the current equation number.
\setcounter{equation}{\value{MYtempeqncnt}}
% IEEE uses as a separator
\hrulefill
% The spacer can be tweaked to stop underfull vboxes.
\vspace*{4pt}
\end{figure*}

\subsection{Proof of Theorem \ref{thm:33}}
To prove Theorem \ref{thm:33}, we consider the right part of  \eqref{bound1}. Recall the KKT conditions \eqref{KKT1} and \eqref{KKT2}:
\begin{align}
\frac{\mu_{1}}{2} \mathbf{C}_{1} & =\lambda \mathbf{C}_{1}^2 -  \mathbf{\Delta}_{1}, \nonumber \\
\frac{\mu_{i}-\mu_{i+1}}{2}\mathbf{C}_{i+1}  & = \mathbf{\Delta}_{i+1} - \mathbf{\Delta}_{i}, \; i=1,\ldots,L-1. \nonumber
\end{align}
By using the spectral decomposition property \ref{con1} of  $\mathbf{C}_{i} =( \mathbf{K}^{-1}+ \sum_{j=i}^{L}\mathbf{B}_{j}^{*})^{-1},\  i=1,2,\ldots, L,$
in Theorem \ref{spectral}, we obtain that
\begin{equation} \label{bound2}
2(1-\gamma)g'(\gamma) \leq I_{1} + I_{2} + I_{3} + I_{4},
\end{equation}
where the terms in the r.h.s are defined at the top of this page.
\addtocounter{equation}{1}

\par

\smallskip
In what follows, we estimate the above four terms respectively, starting with $I_2$.\smallskip
\begin{lemma}
The term $I_{2}$ can be upper bounded by
\begin{equation}\label{bound3}
I_{2} \leq I_{5} + I_{6},
\end{equation}
where
\begin{subequations}
\begin{align}
 I_{5} = &  \sum_{i=1}^{L-1} \tr \Big\{  2 \mathbf{W}_{i}\mathbf{W}_{i}^{T}\mathbf{\Delta}_{i+1}  \mathbf{W}_{i}\mathbf{W}_{i}^{T}\nonumber \\
 & \qquad \cdot \Big(J (\mathbf{X}_{i+1,\gamma}| M_{i+1}, \ldots, M_{L}) - \mathbf{C}_{i+1}^{-1}\Big) \Big\} \label{subb:1} \\
 I_{6} = & \tr \Big\{  2 \mathbf{V}_{L}\mathbf{V}_{L}^{T}\mathbf{\Delta}_{L} \mathbf{V}_{L}\mathbf{V}_{L}^{T} \Big( J (\mathbf{X}_{L,\gamma}| M_{L})-\mathbf{C}_{L}^{-1}) \Big)\Big\}.\label{subb:2}
\end{align}
\end{subequations}
\end{lemma}
\begin{proof}
By Proposition \ref{con3} in Theorem \ref{spectral}:
\begin{align}
\mathbf{V}_{i}\mathbf{V}_{i}^{T} \mathbf{\Delta}_{i+1}  \mathbf{V}_{i}\mathbf{V}_{i}^{T}\qquad\qquad\qquad\qquad\qquad\qquad\qquad \nonumber \\ =\mathbf{W}_{i}\mathbf{W}_{i}^{T} \mathbf{\Delta}_{i+1}  \mathbf{W}_{i}\mathbf{W}_{i}^{T}+  \mathbf{V}_{i+1}\mathbf{V}_{i+1}^{T} \mathbf{\Delta}_{i+1}  \mathbf{V}_{i+1}\mathbf{V}_{i+1}^{T}, \nonumber \\
i=1,\ldots,L. \nonumber
\end{align}
we can rewrite $I_{2}$ as follows:
\begin{subequations}
\begin{align}
& I_{2} \nonumber \\
 = &  \sum_{i=1}^{L-1} \tr \Big\{  2 \mathbf{W}_{i}\mathbf{W}_{i}^{T}\mathbf{\Delta}_{i+1}  \mathbf{W}_{i}\mathbf{W}_{i}^{T}\nonumber \\
 & \qquad \cdot \Big(J (\mathbf{X}_{i+1,\gamma}| M_{i+1}, \ldots, M_{L}) - \mathbf{C}_{i+1}^{-1}\Big) \Big\} \label{quu:2} \\
  &+ \tr \Big\{  2 \mathbf{V}_{L}\mathbf{V}_{L}^{T}\mathbf{\Delta}_{L} \mathbf{V}_{L}\mathbf{V}_{L}^{T} \Big( J (\mathbf{X}_{L,\gamma}| M_{L})-\mathbf{C}_{L}^{-1}) \Big)\Big\}\label{quu:3}\\
& +  \sum_{i=1}^{L-1} \tr \Big\{  2 \mathbf{V}_{i}\mathbf{V}_{i}^{T}\mathbf{\Delta}_{i}  \mathbf{V}_{i}\mathbf{V}_{i}^{T}
  \Big(J (\mathbf{X}_{i,\gamma}| M_{i}, \ldots, M_{L}) - \mathbf{C}_{i}^{-1} \nonumber \\
 & \qquad -J (\mathbf{X}_{i+1,\gamma}| M_{i+1}, \ldots, M_{L}) + \mathbf{C}_{i+1}^{-1}\Big) \Big\} \label{quu:1}\\
 \leq & I_{5} + I_{6}, \nonumber
\end{align}
\end{subequations}
where the last inequality is because \eqref{quu:1} is upper bounded by $0$ as shown below. \par
By definition \eqref{eq:dembo},
\begin{align}
\mathbf{X}_{i,\gamma} &= \sqrt{1-\gamma} \mathbf{X} + \sqrt{\gamma} \mathbf{X}_{\{i,\ldots,L\}}^{G}, \nonumber \\
\mathbf{X}_{i+1,\gamma} &= \sqrt{1-\gamma} \mathbf{X} + \sqrt{\gamma} \mathbf{X}_{\{i+1,\ldots,L\}}^{G}, \nonumber
\end{align}
 where the covariance matrices of $\mathbf{X}_{\{i,\ldots,L\}}^{G}$ and $\mathbf{X}_{\{i+1,\ldots,L\}}^{G}$ are $\mathbf{C}_i=( \mathbf{K}^{-1}+ \sum_{j=i}^{L}\mathbf{B}_{j}^{*})^{-1}$
and $\mathbf{C}_{i+1}=( \mathbf{K}^{-1}+ \sum_{j=i+1}^{L}\mathbf{B}_{j}^{*})^{-1}$ respectively.
In view of the positive semidefinite partial order $$ \mathbf{C}_{i}  \preceq  \mathbf{C}_{i+1},$$
we can assume that
$$\mathbf{X}_{i+1,\gamma} \leftrightarrow \mathbf{X}_{i,\gamma}  \leftrightarrow (M_{i}, M_{i+1},\ldots, M_{L})  \leftrightarrow(M_{i+1}, \ldots, M_{L})$$
form a Markov chain.
Thus by the data processing inequality in Lemma \ref{lea6}, we have
\begin{equation}\label{ge1}
J(\mathbf{X}_{i+1,\gamma} | M_{i+1}, \ldots, M_{L}) \preceq J(\mathbf{X}_{i,\gamma} | M_{i},M_{i+1}, \ldots, M_{L}).
\end{equation}
On the other hand,  $\mathbf{B}_{i}^{*} \mathbf{V}_{i} = 0$ (Proposition \ref{con5} in Theorem \ref{spectral}) yields that
\begin{equation} \label{ge2}
\mathbf{V}_{i}^{T}\mathbf{C}_{i}^{-1} \mathbf{V}_{i}=\mathbf{V}_{i}^{T}\mathbf{C}_{i+1}^{-1} \mathbf{V}_{i},
\end{equation}
and Proposition \ref{con4} in Theorem \ref{spectral} implies that
\begin{equation} \label{ge3}
\mathbf{V}_{i}\mathbf{V}_{i}^{T}\mathbf{\Delta}_{i}  \mathbf{V}_{i}\mathbf{V}_{i}^{T} \preceq 0.
\end{equation}
Finally, combining \eqref{ge1}, \eqref{ge2} and \eqref{ge3} gives the upper bound \eqref{bound3}.
%\begin{align}
% I_{3} \leq & -  \tr \Big\{  \mu_{2} \mathbf{V}_{1}\mathbf{V}_{1}^{T}\mathbf{C}_2 \mathbf{V}_{1}\mathbf{V}_{1}^{T}  \nonumber \\
% & \qquad \cdot \Big( J (\mathbf{X}_{2,\gamma}| M_{2}, \ldots, M_{L})-\mathbf{C}_2^{-1}  \Big)\Big\} \nonumber \\
% =  & I_{5} + I_{6}.
%\end{align}
\par
\end{proof}
\smallskip

Substituting the upper bound \eqref{bound3} into \eqref{bound2} yields
\begin{equation} \label{step2}
2(1-\gamma)g'(\gamma) \leq I_{1}  + I_{5} + I_{6} + I_{3} + I_{4}.
\end{equation}
We now upper bound the first two terms in r.h.s of \eqref{step2}.

%\begin{subequations}\label{step2}
%\begin{align}
%& 2(1-\gamma)g'(\gamma) \nonumber \\
%\leq & \sum_{i=1}^{L-1}\tr \Big\{  (\mu_{i}-\mu_{i+1})\mathbf{U}_{i}\mathbf{U}_{i}^{T}\mathbf{C}_{i+1}  \mathbf{U}_{i}\mathbf{U}_{i}^{T}\nonumber \\
% & \qquad \cdot \Big(J (\mathbf{X}_{i+1,\gamma}| M_{i+1}, \ldots, M_{L}) - \mathbf{C}_{i+1}^{-1} \Big) \Big\} \label{1st} \\
% & \quad -  \tr \Big\{  \mu_{1}\mathbf{U}_{1}\mathbf{U}_{1}^{T}\mathbf{C}_{1}  \mathbf{U}_{1}\mathbf{U}_{1}^{T} \nonumber \\
% & \qquad \cdot \Big( J (\mathbf{X}_{1,\gamma}| M_{1}, \ldots, M_{L})-\mathbf{C}_{1}^{-1}  \Big)\Big\} \label{2nd} \\
% & -  \sum_{i=1}^{L-1} \tr \Big\{  \mu_{i+1}\mathbf{W}_{i}\mathbf{W}_{i}^{T}\mathbf{C}_{i+1}  \mathbf{W}_{i}\mathbf{W}_{i}^{T}\nonumber \\
% & \qquad \cdot \Big(J (\mathbf{X}_{i+1,\gamma}| M_{i+1}, \ldots, M_{L}) - \mathbf{C}_{i+1}^{-1} \Big) \Big\} \label{3rd} \\
% & - \tr \Big\{  \mu_{L} \mathbf{V}_{L}\mathbf{V}_{L}^{T}\mathbf{C}_{L} \mathbf{V}_{L}\mathbf{V}_{L}^{T}  \nonumber \\
% & \qquad \cdot \Big( J (\mathbf{X}_{L,\gamma}| M_{L})-\mathbf{C}_{L}^{-1}  \Big)\Big\} \label{4th}\\
% & \quad - \sum_{i=1}^{L} \tr \Big\{ \mu_{i} ( \mathbf{\Sigma}_{i}^{-1} - \mathbf{B}_{i}^{*}  )^{-1}\nonumber  \\
% &\qquad \cdot \Big(  ( \mathbf{\Sigma}_{i}^{-1} -  \mathbf{B}_{i}^{*} ) - \mathbf{\Sigma}_{i}^{-1} \cov (\mathbf{Y}_{i,\gamma}|\mathbf{X}, M_{i},Q) \mathbf{\Sigma}_{i}^{-1} \Big)\Big\}. \label{5th}
%\end{align}
%\end{subequations}
%\par
%\smallskip

\begin{lemma}
For the terms $I_{1}$ and $I_{5}$,
\begin{equation} \label{1st}
I_{1} + I_{5} \leq I_{7},
\end{equation}
where
\begin{align}\label{b_1st}
 I_{7} = & \sum_{i=1}^{L} \tr \Big\{  2 \mathbf{U}_{i}\mathbf{U}_{i}^{T}\mathbf{\Delta}_{i} \mathbf{U}_{i}\mathbf{U}_{i}^{T} \nonumber \\
&\quad \cdot \Big(  ( \mathbf{\Sigma}_{i}^{-1} -  \mathbf{B}_{i}^{*} ) - \mathbf{\Sigma}_{i}^{-1} \cov (\mathbf{Y}_{i,\gamma}|\mathbf{X}, M_{i},Q) \mathbf{\Sigma}_{i}^{-1} \Big)\Big\}.
\end{align}
\end{lemma}
\smallskip

\begin{proof}
It follows from Proposition \ref{con4} in Theorem \ref{spectral} that
\begin{equation}\label{comp0}
\mathbf{W}_{i}^{T} \mathbf{\Delta}_{i+1} \mathbf{W}_{i} \succ 0,  \qquad i=1, \ldots,L-1.
\end{equation}
On the other hand,
\begin{align} \label{comp}
&J (\mathbf{X}_{i+1,\gamma}| M_{i+1}, \ldots, M_{L}) - \mathbf{C}_{i+1}^{-1}  \nonumber \\
\overset{(a)}\preceq & (1-\gamma) J(\mathbf{X}|M_{i+1}, \ldots, M_{L}) - (1-\gamma)\mathbf{C}_{i+1}^{-1}  \nonumber \\
\overset{(b)} \preceq  & (1-\gamma) \left(\sum_{j=i+1}^{L}\left(\mathbf{\Sigma}_{j}^{-1} - \mathbf{\Sigma}_{j}^{-1} \cov( \mathbf{Y}_{j}| \mathbf{X}, M_{j},Q)\mathbf{\Sigma}_{j}^{-1}\right)\right) \nonumber \\
& \quad - (1-\gamma)\left(\mathbf{K}^{-1}+\mathbf{C}_{i+1}^{-1}\right)  \nonumber \\
\overset{(c)} = &   \sum_{j=i+1}^{L} ( \mathbf{\Sigma}_{j}^{-1} -  \mathbf{B}_{j}^{*} ) - \mathbf{\Sigma}_{j}^{-1} \cov (\mathbf{Y}_{j,\gamma}|\mathbf{X}, M_{j},Q) \mathbf{\Sigma}_{j}^{-1},
\end{align}
where (a) follows from the definition of random vector $\{\mathbf{X}_{{i+1},\gamma}\}$ and Fisher information inequality in Lemma \ref{lea4},
 (b)  can proved by using the argument in \cite[Section 6.2]{v4} (for completeness, we rewrite the proof in \cite{v4} in Appendix \ref{app}), and (c) is due to the definition of  random vector $\{\mathbf{Y}_{j,\gamma}\}$.
\par

Finally, we obtain the bound \eqref{1st} by substituting \eqref{comp0} \eqref{comp} into $I_{5}$ and \eqref{comp} into $I_{1}$ then simplifying it using the relationship ( Proposition \ref{con3} in Theorem \ref{spectral} ):
\begin{align}
\mathbf{U}_{i+1}\mathbf{U}_{i+1}^{T} \mathbf{\Delta}_{i+1}  \mathbf{U}_{i}\mathbf{U}_{i}^{T}\qquad\qquad\qquad\qquad\qquad\qquad\qquad \nonumber \\ =\mathbf{U}_{i}\mathbf{U}_{i}^{T} \mathbf{\Delta}_{i+1}  \mathbf{U}_{i}\mathbf{U}_{i}^{T}+  \mathbf{W}_{i}\mathbf{W}_{i}^{T} \mathbf{\Delta}_{i+1}  \mathbf{W}_{i}\mathbf{W}_{i}^{T}, \nonumber \\
i=1,\ldots,L. \nonumber
\end{align}
\end{proof}
\smallskip

Substituting the upper bound \eqref{1st} into \eqref{step2} gives
\begin{equation} \label{step3}
2(1-\gamma)g'(\gamma) \leq I_{6}+ I_{7} + I_{3} + I_{4}.
\end{equation}
We now upper bound each term separately.

\begin{lemma}
For the first term $I_{6}$ in \eqref{step3},
\begin{equation} \label{4th}
I_{6} \leq 0.
\end{equation}
\end{lemma}

\begin{proof}
By data processing inequality in Lemma \ref{lea6},
\begin{align}
& \mathbf{V}_{L}^{T}J (\mathbf{X}_{L,\gamma}| M_{L})\mathbf{V}_{L} -\mathbf{V}_{L}^{T} ( \mathbf{K}^{-1}+ \mathbf{B}_{L}^{*})\mathbf{V}_{L} \nonumber \\
\succeq & \mathbf{V}_{L}^{T}J (\mathbf{X}_{L,\gamma})\mathbf{V}_{L} -\mathbf{V}_{L}^{T} ( \mathbf{K}^{-1}+ \mathbf{B}_{L}^{*})\mathbf{V}_{L} \nonumber \\
= & \mathbf{V}_{L}^{T} J (\sqrt{1-\gamma} \mathbf{X} + \sqrt{\gamma} \mathbf{X}_{L}^{G})\mathbf{V}_{L} -\mathbf{V}_{L}^{T} ( \mathbf{K}^{-1}+ \mathbf{B}_{L}^{*})\mathbf{V}_{L}  \nonumber \\
= & \mathbf{V}_{L}^{T} ((1-\gamma) \mathbf{K} + \gamma ( \mathbf{K}^{-1}+ \mathbf{B}_{L}^{*})^{-1})^{-1}\mathbf{V}_{L} \nonumber \\
& \quad -\mathbf{V}_{L}^{T} ( \mathbf{K}^{-1}+ \mathbf{B}_{L}^{*})\mathbf{V}_{L} \nonumber \\
= & \mathbf{V}_{L}^{T}\mathbf{K}^{-1} (\mathbf{K}^{-1} + (1-\gamma) \mathbf{B}_{L}^{*})^{-1} (\mathbf{K}^{-1} +  \mathbf{B}_{L}^{*})\mathbf{V}_{L} \nonumber \\
& \quad -\mathbf{V}_{L}^{T} ( \mathbf{K}^{-1}+ \mathbf{B}_{L}^{*})\mathbf{V}_{L} \nonumber \\
=&  0,
\end{align}
where the last step comes from  $\mathbf{B}_{L}^{*} \mathbf{V}_{L} = 0$ in Proposition \ref{con5} in Theorem \ref{spectral}.\par
On the other hand, by Proposition \ref{con4} in Theorem \ref{spectral}, $\mathbf{V}_{L} ^{T} \mathbf{\Delta}_{L} \mathbf{V}_{L}  \preceq 0$, we see that
\begin{align}
&  \tr \Big\{ \mathbf{V}_{L}\mathbf{V}_{L}^{T}\mathbf{\Delta}_{L} \mathbf{V}_{L}\mathbf{V}_{L}^{T} \Big( J (\mathbf{X}_{L,\gamma}| M_{L})-( \mathbf{K}^{-1}+ \mathbf{B}_{L}^{*}) \Big)\Big\}\nonumber \\
\leq & \tr \Big\{ \mathbf{V}_{L}^{T}\mathbf{\Delta}_{L} \mathbf{V}_{L} \cdot \mathbf{0} \Big\} =0. \nonumber
\end{align}
\end{proof}
\smallskip

\begin{lemma}
For the second term $I_{7}$  and the third term $I_{3}$ in \eqref{step3},
\begin{equation}
I_{7}+I_{3} \leq 0.
\end{equation}
\end{lemma}
\smallskip

\begin{proof}
By the definition of $\mathbf{\Delta}_{i}$:
$$\mathbf{\Delta}_{i} \triangleq  \frac{\mu_{i}}{2}( \mathbf{\Sigma}_{i}^{-1} -  \mathbf{B}_{i}^{*})^{-1}-\mathbf{\Psi}_{i},$$
we can write $I_{7}+I_{3}$ in the following form:
\begin{align} \label{bd1}
& I_{7} + I_{3} \nonumber \\
= & -\sum_{i=1}^{L} \tr \Big\{  2 \mathbf{U}_{i}\mathbf{U}_{i}^{T}\mathbf{\Psi}_{i} \mathbf{U}_{i}\mathbf{U}_{i}^{T} \nonumber \\
&\qquad \cdot \Big(  ( \mathbf{\Sigma}_{i}^{-1} -  \mathbf{B}_{i}^{*} ) - \mathbf{\Sigma}_{i}^{-1} \cov (\mathbf{Y}_{i,\gamma}|\mathbf{X}, M_{i},Q) \mathbf{\Sigma}_{i}^{-1} \Big)\Big\}
\end{align}
Considering that
\begin{align} \label{bd2}
 & \cov (\mathbf{Y}_{i,\gamma}|\mathbf{X}, M_{i},Q) \nonumber \\
 \overset{(a)}= &  (1-\gamma)\cov (\mathbf{Y}_{i}|\mathbf{X}, M_{i},Q)  + \gamma ( \mathbf{\Sigma}_{i}-  \mathbf{\Sigma}_{i} \mathbf{B}_{i}^{*}  \mathbf{\Sigma}_{i}) \nonumber \\
 \preceq & (1-\gamma)\cov (\mathbf{Y}_{i}|\mathbf{X}) + \gamma ( \mathbf{\Sigma}_{i}-  \mathbf{\Sigma}_{i} \mathbf{B}_{i}^{*}  \mathbf{\Sigma}_{i}) \nonumber \\
 = &   \mathbf{\Sigma}_{i}- \gamma \mathbf{\Sigma}_{i} \mathbf{B}_{i}^{*}  \mathbf{\Sigma}_{i},
\end{align}
in which  (a) is from the definition of random vector $\{\mathbf{Y}_{i,\gamma}\}$ in Section IV, we have
\begin{align} \label{bd3}
&I_{7} + I_{3} \nonumber \\
\leq & - \tr \Big\{  2 \mathbf{U}_{i}\mathbf{U}_{i}^{T}\mathbf{\Psi}_{i} \mathbf{U}_{i}\mathbf{U}_{i}^{T} \Big(  ( \mathbf{\Sigma}_{i}^{-1} -  \mathbf{B}_{i}^{*} ) - ( \mathbf{\Sigma}_{i}^{-1} -  \gamma \mathbf{B}_{i}^{*} ) \Big)\Big\} \nonumber \\
= & \tr \Big\{  2 \mathbf{U}_{i}^{T}\mathbf{\Psi}_{i} \mathbf{U}_{i} \mathbf{U}_{i}^{T} (1-\gamma) \mathbf{B}_{i}^{*}\mathbf{U}_{i}  \Big\} \nonumber \\
\overset{(a)} = & 2(1-\gamma) \tr \Big\{   \mathbf{U}_{i}^{T}\mathbf{\Psi}_{i} \mathbf{U}_{i} \mathbf{U}_{i}^{T}  \mathbf{B}_{i}^{*}\mathbf{U}_{i} +\mathbf{V}_{i}^{T}\mathbf{\Psi}_{i} \mathbf{V}_{i} \mathbf{V}_{i}^{T}  \mathbf{B}_{i}^{*}\mathbf{V}_{i}  \Big\} \nonumber \\
= & 2(1-\gamma)\tr \Big\{  \mathbf{\Psi}_{i} \mathbf{U}_{i} \mathbf{U}_{i}^{T}  \mathbf{B}_{i}^{*} \Big\} \nonumber \\
= & 2(1-\gamma)\tr \Big\{  \mathbf{U}_{i} \mathbf{U}_{i}^{T}  \mathbf{B}_{i}^{*}\mathbf{\Psi}_{i}  \Big\} \nonumber \\
\overset{(b)} = & 0£¬
\end{align}
where (a) is from $ \mathbf{B}_{i}\mathbf{V}_{i}=0$ of Proposition \ref{con4} in Theorem \ref{spectral}, (b) is from complementary slackness conditions in KKT conditions \eqref{KKT3}: $ \mathbf{B}_{i}^{*} \mathbf{\Psi}_{i} =0$.
\par

 \end{proof}
\smallskip

\begin{lemma}
For the last term $I_{4}$ in \eqref{step3}, we have
\begin{equation}
I_{4} \leq 0.
\end{equation}
\end{lemma}
\smallskip
\begin{proof}
Due to the spectral decomposition of $\mathbf{C}_{1}$:
$$
\mathbf{C}_1= \sum_{n=1}^{m} d_{n} \bm{e}_{n}\bm{ e} _{n}^{T},$$
we see that
\begin{align}\label{sca}
  & - I_{4}/2\lambda \nonumber \\
= &\sum_{n=1}^{m} d_{n}^{2} \tr \{\bm{e}_{n}^{T} J (\mathbf{X}_{1,\gamma}| M_{1}, \ldots, M_{L})\bm{e}_{n} - d_{n}^{-1}\} \nonumber \\
\geq &\sum_{n=1}^{n_{1}} d_{n}^{2}  J (\bm{e}_{n}^{T}\mathbf{X}_{1,\gamma}| M_{1}, \ldots, M_{L})  - \sum_{n=1}^{m} d_{n},
\end{align}
where the inequality in \eqref{sca} is from \cite[Corollary 1-b]{lemma2}:
 $J(\mathbf{\Lambda}\mathbf{N}) \preceq \mathbf{\Lambda}^{T} J(N)\mathbf{\Lambda} $ for any column orthogonal matrix $\mathbf{\Lambda}$.
\par
Let
$$c_{n} \triangleq \cov (\bm{e}_{n}^{T}\mathbf{X} |M_{1}, M_{2}, \ldots, M_{L} ).$$ By the definition of $\{\mathbf{X}_{i,\gamma}\}$ and the Cram\'{e}r--Rao lower bound
in Lemma \ref{lea1},
\begin{align}
&J (\bm{e}_{n}^{T}\mathbf{X}_{i,\gamma}|  M_{1}, \ldots, M_{L})^{-1}
\leq \cov(\bm{e}_{n}^{T}\mathbf{X}_{i,\gamma}|  M_{1}, \ldots, M_{L}) \nonumber \\
= &(1-\gamma)c_{n} + \gamma d_{n} \nonumber
\end{align}
To show that \eqref{sca} is lower-bounded by 0 is equivalent to show:
\begin{equation}\label{sca4}
\sum_{n=1}^{m} d_{n} \frac{d_{n}}{(1-\gamma)c_{n} + \gamma d_{n}} \geq \sum_{n=1}^{m} d_{n}.
\end{equation}
According to Corollary \ref{lamd}, we have
$$\tr \left( \mathbf{C}_{1} \right) = \sum_{n=1}^{m} d_{n} = d.$$
Now consider the trace constraint
\begin{align}
&\tr\left\{ \cov (\mathbf{X}|  M_{1}, M_{2}, \ldots, M_{L})\right\} \nonumber \\
= &\tr\left\{ \cov ((\bm{e}_{1}^{T},\bm{e}_{2}^{T},\ldots, \bm{e}_{m}^{T})\mathbf{X}|  M_{1}, M_{2}, \ldots, M_{L})\right\} \nonumber \\
= & \sum_{n=1}^{m}  \cov (\bm{e}_{n}^{T}\mathbf{X}|  M_{1}, M_{2}, \ldots, M_{L}) \nonumber \\
= & \sum_{n=1}^{m}c_{n} \leq d. \nonumber
\end{align}
Since $f(x)=x^{-1}$ is convex, we have $\sum_{n=1}^{m} \alpha_{n} f(x_{n}) \geq f (\sum_{n=1}^{m}\alpha_{n} x_{n})$, where
$\sum_{n=1}^{m}\alpha_{n}=1, \alpha_{n} \geq 0.$

Let
$$\alpha_{n} = \frac{d_{n}}{d}, \qquad x_{n} = \frac{(1-\gamma)c_{n} + \gamma d_{n}}{d_{n}}.$$
It can be seen that
\begin{align}
& \sum_{n=1}^{m} \frac{d_{n}}{d} \frac{d_{n}}{(1-\gamma)c_{n} + \gamma d_{n}}
\geq     \left( \sum_{n=1}^{m}\frac{d_{n}}{d}  \frac{(1-\gamma)c_{n} + \gamma d_{n}}{d_{n}} \right)^{-1} \nonumber\\
= &\frac{d}{(1-\gamma)\sum_{n=1}^{m}c_{n}+\gamma\sum_{n=1}^{m}d_{n} } \geq 1,
\end{align}
which implies \eqref{sca4}. Thus  $I_{4}$ indeed upper-bounded by $0$.
\par
\end{proof}
\smallskip
This completes the proof of Theorem \ref{thm:33} as well as the extremal inequality in Theorem \ref{thm2}.

\subsection{Rate Distortion Region}
Now we proceed to prove Theorem \ref{thm1}, \emph{i.e.} $R(d) \geq R^{BT}(d)$. To this end we need the Wagner-Anantharam single-letter outer bound  \cite {outer} on $\mathcal{R}(d)$.
\begin{theorem} \cite[Theorem 1]{outer}
The rate region $\mathcal{R}(d)$ is contained in the union of rate tuples $(R_{1}, R_{2}, \ldots, R_{L})$ such that \begin{align}
&\sum_{i=j}^{L} R_{j} \nonumber \\
 \geq & I(\mathbf{X} ;M_{1}, \ldots,M_{i} |M_{i+1}, \ldots, M_{L}) + \sum_{i=j}^{L} I(\mathbf{Y}_{j}; M_{j}|\mathbf{X},Q) \nonumber
\end{align}
where the union is over all joint distributions $p(\mathbf{x}, \mathbf{y}_{1}, \ldots, \mathbf{y}_{L}, m_{1}, \ldots, m_{L},q)$, which can be factorized as follows:
\begin{align}
p(\mathbf{x}, \mathbf{y}_{1}, \ldots, \mathbf{y}_{L}, m_{1}, \ldots, m_{L},q) \nonumber \\
=p(\mathbf{x}) p(q) \prod_{i=1}^{L}p(\mathbf{y}_{i}|\mathbf{x})p(m_{i}|\mathbf{y}_{i},q), \nonumber
\end{align}
and $\tr\{\cov(\mathbf{X}|M_1,\ldots,M_L)\}\leq d$.
\end{theorem}

According to this single-letter outer bound, we have
\begin{align}
 & \sum_{i=1}^{L} \mu_{i} R_{i} \nonumber \\
\geq &  \sum_{i=1}^{L-1}(\mu_{i}-\mu_{i+1}) I(\mathbf{X};M_{1}, \ldots,M_{i} |M_{i+1}, \ldots, M_{L} ) \nonumber \\
& + \mu_{L} I(\mathbf{X};M_{1}, \ldots, M_{L})+ \sum_{i=1}^{L} I(\mathbf{Y}_{i}; M_{i}|\mathbf{X},Q) \nonumber \\
=&\sum_{i=1}^{L-1} (\mu_{i}- \mu_{i+1}) h(\mathbf{X}| M_{i+1}, \ldots, M_{L}) \nonumber \\
& - \mu_{1} h(\mathbf{X}| M_{1}, \ldots, M_{L}) - \sum_{i=1}^{L} \mu_{i} h(\mathbf{Y}_{i} | \mathbf{X}, M_{i},Q) \label{exter} \\
& + \mu_{L}h(\mathbf{\mathbf{X}}) + \sum_{i=1}^{L} h(\mathbf{Y}_{i} | \mathbf{X}). \nonumber
\end{align}
Notice that the term \eqref{exter} equals the l.h.s of extremal inequality \eqref{exinq} in Theorem \ref{thm2}, so that we have
\begin{align}
& R(d)
=   \inf_{(R_{1}, \ldots, R_{L)} \in \mathcal{R}(d)} \sum_{i=1}^{L} \mu_{i} R_{i} \nonumber \\
\geq & \inf_{\tr\{ \cov(\mathbf{X}| M_{1},\ldots, M_{L}) \leq d     \}}\sum_{i=1}^{L-1} (\mu_{i}- \mu_{i+1}) h(\mathbf{X}| M_{i+1}, \ldots, M_{L}) \nonumber \\
& - \mu_{1} h(\mathbf{X}| M_{1}, \ldots, M_{L}) - \sum_{i=1}^{L} \mu_{i} h(\mathbf{Y}_{i} | \mathbf{X}, M_{i},Q) \nonumber \\
& + \mu_{L}h(\mathbf{\mathbf{X}}) + \sum_{i=1}^{L} h(\mathbf{Y}_{i} | \mathbf{X}) \nonumber \\
\geq & \sum_{i=1}^{L-1}\frac{\mu_{i} - \mu_{i+1}}{2} \log |(2\pi e)( \mathbf{K}^{-1}+ \sum_{j=i+1}^{L} \mathbf{B}_{j}^{*})^{-1}| \nonumber \\
&  - \frac{\mu_{1}}{2} \log |(2\pi e)( \mathbf{K}^{-1}+ \sum_{j=1}^{L} \mathbf{B}_{j}^{*})^{-1}| \nonumber \\
&  - \sum_{i=1}^{L}\frac{\mu_{i}}{2} \log |(2\pi e)( \mathbf{\Sigma}_{i} -  \mathbf{\Sigma}_{i} \mathbf{B}_{i}^{*} \mathbf{\Sigma}_{i} )| \nonumber \\
& +   \frac{\mu_{1}}{2} \log |(2\pi e) \mathbf{K} |  +\sum_{i=1}^{L}\frac{\mu_{i}}{2} \log |(2\pi e) \mathbf{\Sigma}_{i} | \nonumber \\
=&  R^{BT}(d). \nonumber
\end{align}
\par
This completes the proof of Theorem \ref{thm1} and establishes the tightness of Berger-Tung inner bound for the vector Gaussian CEO problem with trace distortion constraint.

\section{conclusion}
This paper provides a complete characterization of the rate region of the vector Gaussian CEO problem with the trace distortion constraint. Our proof is based on, among other things, a careful analysis of the KKT conditions for the optimization problem associated with the Berger-Tung scheme. In particular, we exploit the special structure of the KKT conditions to bound the rate region by considering the projection into different subspaces, and the inherent symmetry of the CEO problem enables us to perform the projection procedure recursively.

It should be stressed that the approach in this work does not apply directly to the setting considered in \cite{v3,v4} where a covariance constraint instead of a trace constraint is imposed. However, our work indicates that a more thorough analysis of the KKT conditions might lead to some progress towards that direction.

\appendices

\section{Proof of Lemma \ref{lea5}} \label{OP}

Note that
\begin{align*}
&\gamma^{2} \cov (\mathbf{X}|U) + (1-\gamma)^{2} \mathbf{\Sigma}  \nonumber \\
&\overset{(a)}\succeq (\gamma(\gamma\cov (\mathbf{X}|U))^{-1} + (1-\gamma)((1-\gamma)\mathbf{\Sigma})^{-1})^{-1}\\
&=(\cov (\mathbf{X}|U)^{-1} + \mathbf{\Sigma}^{-1})^{-1},
\end{align*}
which (a) is because $\mathbf{A}^{-1}$ is matrix concave in $\mathbf{A}$.  This together with the fact (see, e.g., \cite[footnote 2]{v2})
\begin{align*}
(\cov (\mathbf{X}|U)^{-1} + \mathbf{\Sigma}^{-1})^{-1}\succeq \cov (\mathbf{X}|\mathbf{X}+\mathbf{N}, U)
\end{align*}
completes the proof of Lemma \ref{lea5}.

%\begin{align}
%&\gamma^{2} \cov (\mathbf{X}|U) + (1-\gamma)^{2} \mathbf{\Sigma}  \nonumber \\
%\succeq &(\cov (\mathbf{X}|U)^{-1} + \mathbf{\Sigma}^{-1})^{-1} \nonumber \\
%= &\cov (\mathbf{X}|U)  (\cov(\mathbf{X}|U) + \mathbf{\Sigma})^{-1}  \mathbf{\Sigma} \nonumber \\
%= &\mathbf{\Sigma} - \mathbf{\Sigma} (\cov(\mathbf{X}|U) + \mathbf{\Sigma})^{-1}  \mathbf{\Sigma} \nonumber \\
%\overset{(a)}\succeq &\mathbf{\Sigma} - \mathbf{\Sigma} J(\mathbf{X}+\mathbf{N}|U) \mathbf{\Sigma} \nonumber \\
%\overset{(b)}=&  \cov (\mathbf{X}|\mathbf{X}+\mathbf{N}, U). \nonumber
%\end{align}
%The inequality (a) is from  Cram\'{e}r-Rao inequality (\ref{lemma1}) in Lemma \ref{lea1} , and (b) from  complementary identity (\ref{lemma2}) in Lemma \ref{lea2}.
%\end{proof}

\section{Existense of KKT Conditions for $R^{BT}(d)$} \label{KKT}
The proof is similar to those in \cite[Appendix IV]{MIMO} and \cite[Appendix B]{vector}. One can refer to  \cite[Sections 4-5]{book} for the background materials.
We first rewrite the optimization problem $R^{BT}(d)$ in a general form:
\begin{align}
\mathop\text{min}_{\bm{b}} &\quad f(\bm{b}) \nonumber \\
\text {subject to} & \quad g(\bm{b}) \leq 0, \nonumber \\
& \quad \bm{b} \in \mathcal{B} \triangleq \mathcal{B}_{1} \cap \mathcal{B}_{2} \cap \ldots \cap \mathcal{B}_{L}.
\end{align}
The vector $\bm{b} \in \mathbb{R}^{Lm^{2} \times 1}$ is constructed by concatenating the columns of $m \times m$ matricies $\mathbf{B}_{1}$ through $\mathbf{B}_{L}$; moreover,
\begin{align}
f(\bm{b}) \triangleq  & \sum_{i=1}^{L-1} \frac{\mu_{i}-\mu_{i+1}}{2} \log \frac{|\mathbf{K}^{-1}+\sum_{j=1}^{L}\mathbf{B}_{j}|}{|\mathbf{K}^{-1}+\sum_{j=i+1}^{L} \mathbf{B}_{j}|} \nonumber \\
& + \sum_{i=1}^{L} \frac{\mu_{i}}{2} \log{\frac{ | \mathbf{\Sigma}_{i}^{-1}| }{ |\mathbf{\Sigma}_{i}^{-1}-\mathbf{B}_{i}| }} + \frac{\mu_{L}}{2} \log \frac{|\mathbf{K}^{-1}+\sum_{j=1}^{L}\mathbf{B}_{j}|}{|\mathbf{K}^{-1}|}, \nonumber
\end{align}
 $$ g(\bm{b}) \triangleq \tr\{(\mathbf{K}^{-1}+ \sum_{i=1}^{L} \mathbf{B}_{i})^{-1} \}-d,$$
and
\begin{align}
\mathcal{B}_{i} \triangleq  \left\{{\text{column concatenation of } (\mathbf{B}_{1}, \mathbf{B}_{2}, \ldots, \mathbf{B}_{L}):  \mathbf{B}_{i} \succeq 0} \right\}, \nonumber \\
i=1,2,\ldots,L.\nonumber
\end{align}
\par
Since $f$ and $g$ are continuously differentiable, the Fritz-John necessary conditions  \cite[Definition 5.2.1]{book} hold: there exist $\mu ,  \lambda \geq 0$
for the local minima $\bm{b}^{*}$ such that
\begin{equation} \label{inclu}
- \left(  \mu \nabla f(\bm{b}^{*}) + \lambda\nabla g(\bm{b}^{*})      \right) \in T_{\mathcal{B}}(\bm{b}^{*})^{*},
\end{equation}
where $T_{\mathcal{B}}(\bm{b}^{*})$ is the \emph {tangent cone} of $\mathcal{B}$ at $\bm{b}^{*}$ and $ T_{\mathcal{B}}(\bm{b}^{*})^{*}$ is its \emph{polar cone}.
\par
As $\mathcal{B}_{i}, i=1,2,\ldots, L$ are nonempty convex sets such that $\ri(\bm{b}_{1}^{*})^{*} \cap \ri(\bm{b}_{2}^{*})^{*} \cap \cdots \cap  \ri(\bm{b}_{L}^{*})^{*}$
is nonempty, it follows \cite[Problem 4.23]{book} and \cite[Proposition 4.63]{book} that
$$
 T_{\mathcal{B}}(\bm{b}^{*})^{*} = T_{\mathcal{B}_{1}}(\bm{b}^{*})^{*} + T_{\mathcal{B}_{2}}(\bm{b}^{*})^{*} + \cdots + T_{\mathcal{B}_{L}}(\bm{b}^{*})^{*}.
$$
As in \cite[Section B]{vector}, it can be verified that
\begin{align}
T_{\mathcal{B}_{i}}(\bm{b}^{*})^{*} \cap \mathcal{A} \subseteq  & \left\{ \text{column concatenation of } (\mathbf{O}, \ldots, -\mathbf{\Psi}_{i},\right. \nonumber \\
 &  \left.  \ldots, \mathbf{O}): \mathbf{\Psi}_{i} \succeq 0, \tr \{\mathbf{\Psi}_{i} \mathbf{B}_{i}^{*}\} =0  \right\}
\end{align}
in which $\mathcal{A}$ is the set of vectors constructed by concatenating the columns of $L$ symmetric matrices.
\par
Since l.h.s of equation \eqref{inclu} is also in $\mathcal{A}$, to complete the proof of the existence of KKT conditions, we need to show $\mu\neq 0$. As in \cite[Appendix IV]{MIMO},
we will verify the constraint qualifications (CQ5a in \cite[Section 5.4]{book}), i.e., there exists a vector
$$\bm{d} \in T_{\mathcal{B}}(\bm{b}^{*}) = T_{\mathcal{B}_{1}}(\bm{b}^{*})\cap T_{\mathcal{B}_{2}}(\bm{b}^{*}) \cap \cdots \cap T_{\mathcal{B}_{L}}(\bm{b}^{*}),$$
such that $\nabla g(\bm{b}^{*})^{T}  \bm{d} < 0 $.
\par
Given any $\alpha > 1 $, let's define a set of $m^{2} \times 1$ vectors
 \begin{align}
 \bm{b}_{i}= vec \left(\mathbf{B}_{i}\right) \triangleq vec \left(\alpha \mathbf{B}_{i}^{*} + \frac{\alpha-1}{L} \mathbf{K}^{-1}\right),
  i=1,2,\ldots, L.
 \end{align}
Here $vec(\cdot)$ is the vectorization operator. It can be seen that $\bm{b}_{i} \in \mathcal{B}_{i} $ since $\mathbf{B}_{i} \succeq 0$.
We denote $\bm{d}_{i} = \bm{b}_{i} - (\bm{b}^{*})_{i} $, where $(\bm{b}^{*})_{i}$ denotes the $ith$ $L$-components in $\bm{b}^{*}$. By \cite[Definition 4.6.1]{book} and
\cite[Proposition 4.6.2]{book}, we have $\bm{d}_{i} \in T_{\mathcal{B}_{i}}(\bm{b}^{*})$. Now $\bm{d}$ can be constructed by
$$\bm{d} = vec (\bm{d}_{1}, \bm{d}_{2}, \ldots, \bm{d}_{L}).$$
In this way, the expression of $\nabla g(\bm{b}^{*})^{T}  \bm{d}$ can be written as
\begin{align}
&\sum_{i=1}^{L} \tr \left\{ (\mathbf{K}^{-1}+ \sum_{i=1}^{L} \mathbf{B}_{i}^{*})^{-2} \left(   \mathbf{B}_{i}^{*} -\mathbf{B}_{i}   \right)                   \right\} \nonumber \\
=&\sum_{i=1}^{L} \tr \left\{ (\mathbf{K}^{-1}+ \sum_{i=1}^{L} \mathbf{B}_{i}^{*})^{-2}   \left(  (1-\alpha)\mathbf{B}_{i}^{*} - \frac{\alpha-1}{L} \mathbf{K}^{-1}        \right)           \right\} \nonumber \\
= & (1-\alpha)\tr \left\{ (\mathbf{K}^{-1}+ \sum_{i=1}^{L} \mathbf{B}_{i}^{*})^{-1} \right\} \nonumber \\
< & \;0, \nonumber
\end{align}
where the inequality is because $1-\alpha <0 $ and $(\mathbf{K}^{-1}+ \sum_{i=1}^{L} \mathbf{B}_{i}^{*})^{-1}\succ 0$. This completes the proof of the existence of KKT conditions for the non-convex optimization problem $R^{BT}(d)$.

\section {Proof of Inequality (\textnormal{b}) in \eqref{comp}} \label{app}
We shall  show that
\begin{align}
&  J(\mathbf{X}|M_{i+1}, \ldots, M_{L})  \nonumber \\
\preceq & \mathbf{K}^{-1}+ \sum_{j=i+1}^{L}\left(\mathbf{\Sigma}_{j}^{-1} - \mathbf{\Sigma}_{j}^{-1} \cov( \mathbf{Y}_{j}| \mathbf{X}, M_{j},Q)\mathbf{\Sigma}_{j}^{-1}\right)
\end{align}
Note that
$$\mathbf{X}=\sum_{j=i+1}^{L} \mathbf{A}_{j} \mathbf{Y}_{j}+ \mathbf{Z} \triangleq \bar{\mathbf{X}}+\mathbf{Z},$$
where $\mathbf{Z}$ is a Gaussian random vector, independent of
$(\mathbf{Y}_{i+1}, \ldots, \mathbf{Y}_{L})$, with mean zero and covariance matrix $\mathbf{K}_{\mathbf{Z}} \triangleq(\mathbf{K}^{-1} + \sum_{j=i+1}^{L} \mathbf{\Sigma}_{j}^{-1})^{-1}$, and $\mathbf{A}_{j} \triangleq \mathbf{K}_{\mathbf{Z}} \mathbf{\Sigma}_{j}^{-1}$. Using the
complementary relationship between Fisher information and MSE in Lemma \ref{lea2}, we have
\begin{align}
& J(\mathbf{X}| M_{i+1}, \ldots, M_{L}) \nonumber \\
\overset{(a)}\preceq & J(\mathbf{X}| M_{i+1}, \ldots, M_{L},Q) \nonumber \\
= &J(\bar{\mathbf{X}}+\mathbf{Z}| M_{i+1}, \ldots, M_{L},Q) \nonumber \\
=& \mathbf{K}_{\mathbf{Z}}^{-1} - \mathbf{K}_{\mathbf{Z}}^{-1} \cov(\bar{\mathbf{X}}|\bar{\mathbf{X}}+\mathbf{Z}, M_{i+1}, \ldots, M_{L},Q)\mathbf{K}_{\mathbf{Z}}^{-1} \nonumber \\
\overset{(b)}= & \mathbf{K}_{\mathbf{Z}}^{-1}- \sum_{j=i+1}^{L}\mathbf{\Sigma}_{j}^{-1} \cov (\mathbf{Y}_{j} | \mathbf{X}, M_{j},Q)\mathbf{\Sigma}_{j}^{-1} \nonumber \\
=  & \mathbf{K}^{-1}+ \sum_{j=i+1}^{L}\left(\mathbf{\Sigma}_{j}^{-1} - \mathbf{\Sigma}_{j}^{-1} \cov( \mathbf{Y}_{j}| \mathbf{X}, M_{j},Q)\mathbf{\Sigma}_{j}^{-1}\right),
\end{align}
where (a) is from the data processing inequality in Lemma \ref{lea6} and (b) is due to the fact that for any $j$, the Markov chain $(\mathbf{Y}_{j},M_{j})  \leftrightarrow (\mathbf{X}, Q) \leftrightarrow (\mathbf{Y}_{\{j\}^{c}}, M_{\{j\}^{c}})$ holds.

\bibliographystyle{IEEEtran}
\bibliography{ref}

\begin{thebibliography}{10}
\providecommand{\url}[1]{#1}
\csname url@samestyle\endcsname
\providecommand{\newblock}{\relax}
\providecommand{\bibinfo}[2]{#2}
\providecommand{\BIBentrySTDinterwordspacing}{\spaceskip=0pt\relax}
\providecommand{\BIBentryALTinterwordstretchfactor}{4}
\providecommand{\BIBentryALTinterwordspacing}{\spaceskip=\fontdimen2\font plus
\BIBentryALTinterwordstretchfactor\fontdimen3\font minus
  \fontdimen4\font\relax}
\providecommand{\BIBforeignlanguage}[2]{{%
\expandafter\ifx\csname l@#1\endcsname\relax
\typeout{** WARNING: IEEEtran.bst: No hyphenation pattern has been}%
\typeout{** loaded for the language `#1'. Using the pattern for}%
\typeout{** the default language instead.}%
\else
\language=\csname l@#1\endcsname
\fi
#2}}
\providecommand{\BIBdecl}{\relax}
\BIBdecl

\bibitem{CEO}
T.~Berger, Z.~Zhang, and H.~Viswanathan, ``The {CEO} problem,'' \emph{{IEEE}
  Trans. Inf. Theory}, vol.~42, no.~3, pp. 887 --902, May 1996.

\bibitem{Oohama}
Y.~Oohama, ``The rate-distortion function for the quadratic {G}aussian {CEO}
  problem,'' \emph{{IEEE} Trans. Inf. Theory}, vol.~43, no.~11, pp. 1912
  --1923, Nov. 1997.

\bibitem{regionCEO2}
------, ``Rate-distortion theory for {G}aussian multiterminal source coding
  systems with several side information at the decoder,'' \emph{{IEEE} Trans.
  Inf. Theory}, vol.~51, no.~7, pp. 2577 --2593, Jul. 2005.

\bibitem{regionCEO}
V.~Prabhakaran, D.~Tse, and K.~Ramachandran, ``Rate region of the quadratic
  {G}aussian {CEO} problem,'' in \emph{Proc. {IEEE} Int. Symp. Inf. Theory},
  Jun. 2004, p. 119.

\bibitem{tav}
S.~Tavildar and P.~Viswanath, ``On the sum-rate of the vector {G}aussian {CEO}
  problem,'' in \emph{Proc. 39th Asilomar Conf. on Signal, Syst. and Comput.},
  2005, pp. 3--7.

\bibitem{v1}
G.~Zhang and W.~Kleijn, ``Bounding the rate region of the two-terminal vector
  {G}aussian {CEO} problem,'' in \emph{Proc. Data Comp. Conf.}, Mar. 2011, p.
  488.

\bibitem{v2}
J.~Chen and J.~Wang, ``On the vector {G}aussian {CEO} problem,'' in \emph{Proc.
  {IEEE} Int. Symp. Inf. Theory}, Aug. 2011, pp. 2050 --2054.

\bibitem{v3}
J.~Wang and J.~Chen, ``On the vector {G}aussian {L}-terminal {CEO} problem,''
  in \emph{Proc. {IEEE} Int. Symp. Inf. Theory}, July 2012, pp. 571 --575.

\bibitem{v4}
E.~Ekrem and S.~Ulukus, ``An outer bound for the vector {G}aussian {CEO}
  problem,'' \emph{{IEEE} Trans. Inf. Theory}, submitted for publication.
  [Online]. Available: arXiv:1202.0536.

\bibitem{exinq}
T.~Liu and P.~Viswanath, ``An extremal inequality motivated by multiterminal
  information-theoretic problems,'' \emph{{IEEE} Trans. Inf. Theory}, vol.~53,
  no.~5, pp. 1839 --1851, May 2007.

\bibitem{MIMO}
H.~Weingarten, Y.~Steinberg, and S.~Shamai, ``The capacity region of the
  {G}aussian multiple-input multiple-output broadcast channel,'' \emph{{IEEE}
  Trans. Inf. Theory}, vol.~52, no.~9, pp. 3936 --3964, Sept. 2006.

\bibitem{scalar}
M.~Rahman and A.~Wagner, ``Rate region of the {G}aussian scalar-help-vector
  source-coding problem,'' \emph{{IEEE} Trans. Inf. Theory}, vol.~58, no.~1,
  pp. 172 --188, Jan. 2012.

\bibitem{vector}
------, ``Rate region of the vector {G}aussian one-helper source-coding
  problem,'' \emph{{IEEE} Trans. Inf. Theory}, submitted for publication.
  [Online]. Available: arXiv:1112.6367.

\bibitem{my}
Y.~Xu and Q.~Wang, ``A perturbation proof of the vector {G}aussian
  {One-Help-One} problem,'' in \emph{Proc. {IEEE} Int. Symp. Inf. Theory},
  Istanbul, Turkey, Jul. 2013.

\bibitem{Palomar}
D.~Palomar and S.~Verd\'{u}, ``Gradient of mutual information in linear vector
  {G}aussian channels,'' \emph{{IEEE} Trans. Inf. Theory}, vol.~52, no.~1, pp.
  141--154, 2006.

\bibitem{Dembo}
A.~Dembo, T.~Cover, and J.~Thomas, ``Information theoretic inequalities,''
  \emph{{IEEE} Trans. Inf. Theory}, vol.~37, no.~6, pp. 1501 --1518, Nov. 1991.

\bibitem{lemma2}
R.~Zamir, ``A proof of the {F}isher information inequality via a data
  processing argument,'' \emph{{IEEE} Trans. Inf. Theory}, vol.~44, no.~3, pp.
  1246 --1250, May 1998.

\bibitem{EPI}
O.~Rioul, ``Information theoretic proofs of entropy power inequalities,''
  \emph{{IEEE} Trans. Inf. Theory}, vol.~57, no.~1, pp. 33--55, 2011.

\bibitem{outer}
A.~Wagner and V.~Anantharam, ``An improved outer bound for multiterminal source
  coding,'' \emph{{IEEE} Trans. Inf. Theory}, vol.~54, no.~5, pp. 1919--1937,
  2008.

\bibitem{book}
D.~P. Bertsekas, A.~Nedi{\'c}, and A.~E. Ozdaglar, \emph{Convex Analysis and
  Optimization}.\hskip 1em plus 0.5em minus 0.4em\relax Athena Scientific
  Belmont, 2003.

\end{thebibliography}

\end{document}